\documentclass{acm_proc_article-sp}
\usepackage{graphicx,comment}
\usepackage{amssymb}
\usepackage{enumerate}
\usepackage{multirow}
\usepackage{amsmath}
\usepackage{tabularx}
\usepackage{dsfont}
\usepackage{extarrows}
\usepackage{epstopdf}
\usepackage[]{algorithm2e}
\usepackage{xcolor}
\usepackage{fixltx2e}

\begin{document}

\title{Modelling and Analysis of Network Security\\- a Probabilistic Value-passing CCS Approach}

\numberofauthors{3}
\author{
\alignauthor Qian Zhang\\
       \affaddr{State Key Laboratory of Computer Science,Institute of Software Chinese Academy of Sciences}\\
       \affaddr{Beijing, China}\\
       \email{zhangq@ios.ac.cn}
\alignauthor Ying Jiang\\
       \affaddr{State Key Laboratory of Computer Science, Institute of SoftwareChinese Academy of Sciences}\\
       \affaddr{Beijing, China}\\
       \email{jy@ios.ac.cn}
\alignauthor Liping Ding\\
       \affaddr{National Engineering Research Center for Fundamental Software,Institute of Software Chinese Academy of Sciences}\\
       \affaddr{Beijing, China}\\
       \email{liping@nfs.iscas.ac.cn}
       }
\maketitle

\begin{abstract}
In this work, we propose a probabilistic value-passing CCS (Calculus of Communicating System) approach to model and analyze a typical network security scenario with one attacker and one defender.
By minimizing this model with respect to probabilistic bisimulation and abstracting it through graph-theoretic methods,
two algorithms based on backward induction are designed to compute Nash Equilibrium strategy and Social Optimal strategy respectively. For each algorithm,
the correctness  is proved and an implementation is realized.
Finally, this approach is illustrated by a detailed case study.
\end{abstract}



\section{Introduction}
Modeling and analysis of network security has been a hot research spot in the network security domain.
It has been studied from different perspectives.  Among them are two main approaches, one based on game-theoretic methods \cite{martin},  and
one based on (probabilistic) process algebra \cite{robin,rob,yuxin}.
In the later 1990's, game theoretic methods were introduced for modeling and analyzing network security \cite{syverson}.
These methods consist in applying different kinds of games to different network scenarios with one attacker and one defender \cite{sankardas}.
Roughly speaking,
{\it static game} is a one-shot game in which players choose action simultaneously.
It is often used to model the scenarios in which the attacker and defender have no idea on the action chosen by the adversary (for instance the scenario of information warfare),
and to compute the best strategy for players in a quantitative way \cite{jormakka}.
{\it Stochastic game} is often used to model the scenarios which involve probabilistic transitions through states of network systems according to the actions chosen by the attacker and the defender \cite{nguyen,klye}.
{\it Markov game} is an extension of game theory to MDP-like environments \cite{vander}.
It is often used to model the scenarios in which the future offensive-defensive behaviors will impact on the present action choice of attacker and defender \cite{xiaolin}.
In {\it Bayesian game}, the characteristics about other players is incomplete and players use Bayesian analysis in predicting the outcome \cite{Harsanyi}.
A dynamic Bayesian game with two players, called {\it Signaling game}, is often used to model intrusion detection in mobile ad-hoc networks and to analyze Nash equilibrium in a qualitative way \cite{patcha}.
On the other hand, as far as we know, (probabilistic) process algebra approach focus on verifying network security protocols.
For example, in the earlier 1980's, a simple version of the alternating bit protocol
in $ACP_\tau$ (Algebra of Communicating Processes with silent actions) was verified \cite{bergstra}.
For describing and analyzing cryptographic protocols, the spi calculus, an extension of the $\pi$ calculus, was designed \cite{abadi}.
Recently, a generalization of the bisimilarity pseudo-metric based on the Kantorovich lifting is proposed,
this metric allows to deal with a wider class of properties such as those used in security and privacy \cite{lilixu}.

In this paper, we propose a probabilistic value-passing CCS (PVCCS) approach for modeling and analyzing a typical network security scenario with one attacker and one defender.
A network system is supposed to be composed of three participants: one attacker, one defender and the network environment which is the hardware and software services of the network under consideration. We consider all possible behaviors of the participants at each state of the system as processes and assign each state with a process describing all possible interactions currently performed among the participants.
In this way we establish  a network state transition model, often called reactive model in the literature \cite{rob}, based on PVCCS.
By minimizing this model with respect to probabilistic bisimulation and abstracting it via graph-theoretic methods,
two algorithms based on backward induction are designed to compute Nash Equilibrium Strategy (NES) \cite{xiannuan, jean,ls} and Social Optimal Strategy (SOS) \cite{rd,david} respectively.
The former represents a stable strategy of which neither the attacker nor the defender is willing to change the current situation, and
the latter is the policy to minimize the damages caused by the attacker. For each algorithm, the correctness  is proved and an implementation is realized. This approach is illustrated by a detailed case study on an example introduced in \cite{klye}. The example
describes a local network connected to Internet under the assumption that the firewall is unreliable, and the operating system on the machine is insufficiently hardening, and the attacker has chance to pretend as a root user in web server, stealing or damaging data stored in private file server and private workstation.
The major contributions of our work are:
\begin{itemize}
\item
establish a reactive model based on PVCCS for a typical network security scenario
which is usually modeled via perfect and complete information games.
\item
minimize the state space of network system via probabilistic bisimulation and abstract it via graph-theoretic methods.
This allows us to reduce the search space and hence considerably optimize the complexity of the concerned algorithms.

\item
propose two algorithms to compute Nash Equilibrium and Social Optimal strategy respectively.
The novelty consists in combing graph-theoretic methods with backward induction,  which  enables us on the one hand to increase reuseness and on the other hand to make the backward induction possible in the setting of some infinite paths.
\end{itemize}
Note that our method can filter out invalid Nash Equilibrium strategies from the results obtained by traditional game-theoretic methods. For instance, in the example introduced in \cite{klye},
three Nash Equilibrium strategies obtained ultimately by game-theoretic approach methods,
while only two of them obtained by our method:
we filter out the invalid Nash Equilibrium strategy from the results in \cite{klye}.
Note also our method can be applied to other network security scenarios.
For example, the proposed reactive model can be extended conservatively to a generative model based on PVCCS.
In this way we provide a uniform framework for modeling and analyzing network security scenarios which are usually modeled either via perfect and complete information games or via perfect and incomplete games.
However, for the limited space of this paper, we will focus on the reactive setting for the conciseness and easier understanding of this work.

In the remaining sections, we shall
review some notions of graph theory and establish the reactive model based on PVCCS (Section 2);
present the formal definitions of NES and SOS in this model, as well as the corresponding algorithms
and their correctness proofs (Section 3);
then illustrate our method by a case study (Section 4);
fianlly, discuss the conclusion (Section 5).
Appendix shows proofs of theorems, tables referred to the case study and a notation index.

\section{Preliminaries and Reactive model based on PVCCS}

\subsection{Graph theory}
We firstly recall some notions of graph theory: Strongly Connected Component (SCC), Directed Acyclic Graph (DAG) and Path Contraction \cite{reinhard,serge,david}.

\textbf{SCC} of an arbitrary directed graph form a partition into subgraphs that are themselves strongly connected (it is possible to reach any vertex starting from any other vertex by traversing edges in the direction).

\textbf{DAG} is a directed graph with no directed cycles. There are two useful DAG related properties we used in our paper: (1) if $H$ is a weakly connected graph, $H'$ is obtained by viewing each SCC in $H$ as one vertex, $H'$ must be a DAG; (2) if $H$ is a DAG, $H$ has at least one vertex whose out-degree is 0.

\textbf{Path Contraction}
Let $e=xy$ be an edge of a graph $H=(V,E)$. $H/e$ is a graph $(V',E')$ with vertex set $V':=(V\backslash \{x,y\})\cup \{v_{e}\}$, and edge set $E':=\{vw\in E\mid \{v,w\}\cap \{x,y\}=\varnothing\}\cup \{v_{e}w\mid ~xw\in E\backslash\{e\}~or~yw\in E\backslash \{e\}\}$ (Figure \ref{edgecontraction}). Path contraction occurs upon the set of edges in a path that contract to form a single edge between the endpoints of the path after a series of edge contractions.
\begin{figure}[htpb]
\begin{center}
\includegraphics[scale=0.6]{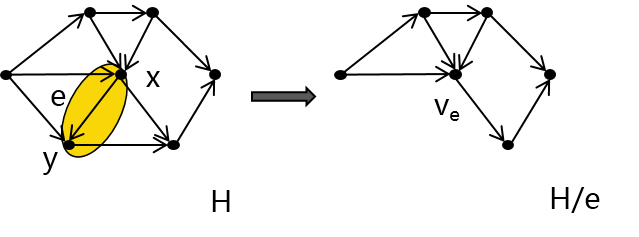}
\caption{Edge contraction}
\label{edgecontraction}
\end{center}
\end{figure}
\subsection{PVCCS\textsubscript{R}}
\newtheorem{definition}{Definition}[section]
\newtheorem{Theorem}{Theorem}[section]
\newtheorem{Lemma}{Lemma}[section]
$\rm{PVCCS_\textit{R}}$ is a reactive model for Probabilistic Value-passing CCS, proposed based on the reactive model for probabilistic CCS \cite{rob}.

\textbf{Syntax:}
Let $\mathbf{\mathcal{A}}$ be a set of channel names ranged over by $a$, and $\mathbf{\mathcal{\overline{A}}}$ be the set of co-names, i.e., $\mathbf{\mathcal{\overline{A}}}=\{\overline{a}\mid a\in \mathbf{\mathcal{A}}\}$, and $\overline{\overline{a}}=a$ by convention. $\mathbf{Label}=\mathcal{A} \cup \mathcal{\overline{A}}$. $\mathbf{Var}$ is a set of value variables ranged over by $x$ and $\mathbf{Val}$ is a value set ranged over by $v$. $\mathbf{e}$ and $\mathbf{b}$ denote value expression and boolean expression respectively. The set of actions, ranged over by $\alpha$, $\mathbf{Act}=\{a(x)\mid a\in \mathcal{A}\}\cup \{\overline{a}(\mathbf{e})\mid \overline{a}\in \mathcal{\overline{A}}\}\cup \{\tau \}$,  where $\tau$ is the silent action. $\mathbf{\mathcal{K}}$ and $\mathbf{\mathcal{X}}$ are a set of process identifiers and a set of process variables respectively. Each process identifier $A\in \mathcal{K}$ is assigned an arity, a non-negative integer representing the number of parameters which it takes.

$\mathbf{Pr}$ is the set of processes in $\rm{PVCCS_{\textit{R}}}$ defined inductively as follows, where $P$, $P_{i}$ are already in $\mathbf{Pr}$:\\
\begin{equation*}
\begin{aligned}
\mathbf{Pr}::=&Nil\mid \underset{i\in I}{\sum}\underset{j\in J}{\sum}[p_{\textit{ij}}]\alpha_{\textit{i}}.P_{\textit{ij}}\mid P_1|P_2\mid P\backslash R\\&
\mid \textit{if}~\mathbf{b}~\textit{then}~P_1~\textit{else}~P_2
\mid A(x)\\
\alpha::=&a(x)\mid \overline{a}(\mathbf{e})
\end{aligned}
\end{equation*}
where $a\in \mathbf{Label}$, $R\subseteq \mathcal{A}$. $I$, $J$ are index sets, and $\forall i\in I$, $p_{\textit{ij}}\in (0,1]$, $\underset{j\in I}{\dot{\sum}}p_{\textit{ij}}=1$, and $\alpha_{\textit{i}}\neq \alpha_{\textit{j}}$ if $i\neq j$. $\sum$ and $\dot{\sum}$ are summation notations for processes and real numbers respectively. Furthermore, each process constant $A(x)$ is defined recursively by associating to each identifier an equation of the form $A(x)\stackrel{\textit{def}}{=}P$, where $P$ contains no process variables and no free value variables except $x$.

$Nil$ is an empty process which does nothing; $\underset{i\in I}{\sum}\underset{j\in J}{\sum}[p_{\textit{ij}}]\alpha_{\textit{i}}.P_{\textit{ij}}$ is a summation process with probabilistic choice which means if performs action $\alpha_{\textit{i}}$, $P_{\textit{ij}}$ will be chosen to be proceed with probability $p_{\textit{ij}}$, for example, $[0.2]\alpha.P_{1}+ [0.8]\alpha.P_{2}+[1]\beta.P_{3}$ is a process which will choose process $P_{1}$ with probability 0.2 and $P_{2}$ with probability 0.8 if performs action $\alpha$, or will choose $P_{3}$ with probability 1 if performs action $\beta$, here $\alpha_{i}$ stands for an action prefix and there are two kinds of prefixes: input prefix $a(x)$ and output prefix $\overline{a}(\mathbf{e})$. If $J$ is a singleton set, then we will omit the probability from the summation process, such as $\underset{i\in I}{\sum}\underset{j\in J}{\sum}[1]\alpha_{\textit{i}}.P_{\textit{ij}}$ will be written as $\underset{i\in I}{\sum}\alpha_{\textit{i}}.P_{\textit{i}}$, and if both $I$ and $J$ are singleton sets, then the summation process is written as $\alpha.P$; $P_{1}|P_{2}$ represents the combined behavior of $P_{1}$ and $P_{2}$ in parallel; $P\backslash R$ is a channel restriction, whose behavior is like that of $P$ as long as $P$ does not perform any action with channel $a\in R\cup \overline{R}$;
$\textit{if}~\mathbf{b}~\textit{then}~P_1~\textit{else}~P_2$ is a conditional process which enacts $P_{1}$ if $\mathbf{b}$ is $true$, else $P_{2}$.

\textbf{Semantics:}
The operational semantics of $\rm{PVCCS_{\textit{R}}}$ is defined by the rules in Table \ref{operationalsemanticspvccs}, where $P\stackrel{\alpha[p]}{\rightarrow}Q$ describes a transition that, by performing an action $\alpha$, starts from $P$ and leads to $Q$ with probability $p$. Mapping $chan:\mathbf{Act}\rightarrow \mathcal{A}$, i.e., $chan(a(x))=chan(\overline{a}(\mathbf{e}))=a$. And $P\{\mathbf{e}/x\}$ means substituting $\mathbf{e}$ for every free occurrences of $x$ in process $P$. By convention, if $P_{1}\stackrel{\alpha[p]}{\rightarrow}P_{2}$ and $P_{2}\stackrel{\beta[q]}{\rightarrow}P_{3}$, then we use $P_{1}\stackrel{\alpha[p]\beta[q]}{\Longrightarrow}P_{3}$ to represent multi-step transition.
\begin{table*}
\centering
\renewcommand{\arraystretch}{2.5}
\begin{tabular}{ll}\hline
$[In]\frac{}{\underset{i\in I}{\sum}\underset{j\in J}{\sum}[p_{\textit{ij}}]a(x).P_{\textit{ij}}\stackrel{a(\mathbf{e})[p_{\textit{ij}}]} {\longrightarrow}P_{\textit{ij}}\{\mathbf{e}/x\}}$&
$[Out]\frac{}{\underset{i\in I}{\sum}\underset{j\in J}{\sum}[p_{\textit{ij}}]\overline{a}(\mathbf{e}).P_{\textit{ij}} \stackrel{\overline{a}(\mathbf{e})[p_{\textit{ij}}]}{\longrightarrow}P_{\textit{ij}}} $\\
$[Res] \frac{P\stackrel{\alpha[p]}{\longrightarrow}P'}{P\backslash R\stackrel{\alpha[p]}{\longrightarrow}P'\backslash R}~~(chan(\alpha)\notin R\cup \overline{R}$) &
$[Con]\frac{P\{\mathbf{e}/x\}\stackrel{\alpha[p]}{\longrightarrow} P'}{A(\mathbf{e})\stackrel{\alpha[p]}{\longrightarrow}P'}~~(A(x)\stackrel{\textit{def}} {=}P)   $\\
$[\textit{Par}_{\textit{l}}] \frac{P_1\stackrel{\alpha[p]}{\longrightarrow}P_1^{'}}{P_1|P_2\stackrel{\alpha [p]}{\longrightarrow}P_1^{'}|P_2}$&
$[\textit{Par}_{\textit{r}}] \frac{P_2\stackrel{\alpha[p]}{\longrightarrow}P_2^{'}}{P_1|P_2\stackrel{\alpha [p]}{\longrightarrow}P_1|P_2^{'}} $\\
$[Com]\frac{P_1\stackrel{a(\mathbf{e})[p]}{\longrightarrow}P_1^{'},~ P_2\stackrel {\overline{a}(\mathbf{e})[q]}{\longrightarrow}P_2^{'}}{P_1|P_2\stackrel{\tau[p\cdot q]}{\longrightarrow}P_1^{'}|P_2^{'}}$ & $ $\\
$[\textit{If}_{\textit{t}}]\frac{P_{1}\stackrel{\alpha[p]}{\longrightarrow}P_{1}^{'}} {\textit{if}~\mathbf{b}~\textit{then}~P_1~\textit{else}~P_2\stackrel{\alpha[p]} {\longrightarrow}P_{1}^{'}}(\mathbf{b}=true)  $ & $ [\textit{If}_{\textit{f}}]\frac{P_{2}\stackrel{\alpha[p]}{\longrightarrow}P_{2}^{'}} {\textit{if}~\mathbf{b}~\textit{then}~P_1~\textit{else}~P_2\stackrel{\alpha[p]} {\longrightarrow}P_{2}^{'}}(\mathbf{b}=\textit{false})  $\\
\hline
\end{tabular}
\caption{\label{operationalsemanticspvccs}Operational semantics of $\rm{PVCCS_{\textit{R}}}$}
\end{table*}

\textbf{Probabilistic Bisimulation:}
We recall the definition of \textsf{cumulative probability distribution function} (cPDF) \cite{rob} which computes the total probability in which a process derives a set of processes. $\wp$ is the powerset operator and we write $\mathbf{Pr}/\mathcal{R}$ to denote the set of equivalence classes induced by equivalence relation $\mathcal{R}$ over $\mathbf{Pr}$.
\begin{definition}
$\mu:(\mathbf{Pr}\times Act\times \wp(\mathbf{Pr}))\rightarrow[0,1]$ is the total function given by: $\forall \alpha\in Act$, $\forall P\in \mathbf{Pr}$, $\forall C\subseteq \mathbf{Pr}$, $\mu(P,\alpha,C)=\dot{\sum}\{p|P\stackrel{\alpha[p]}{\longrightarrow}Q,~Q\in C\}$. \end{definition}

\begin{definition}
An equivalence relation $\mathcal{R}\subseteq \begin{bf}Pr\end{bf} \times \begin{bf}Pr\end{bf}$ is a \textsf{probabilistic bisimulation} if $(P, Q)\in \mathcal{R}$ implies: $\forall C\in \mathbf{Pr}/\mathcal{R}$, $\forall\alpha\in Act$, $\mu(P, \alpha, C)=\mu(Q, \alpha, C)$.
\end{definition}
$P$ and $Q$ are probabilistic bisimilar, written as $P\sim Q$, if there exists a probabilistic bisimulation $\mathcal{R}$ s.t. $P\mathcal{R}Q$.

\subsection{Modelling for Network Security based on  PVCCS\textsubscript{R}}
\textbf{ComModel} focuses on modeling the network security scenario modeled usually via perfect and complete information game: a network system state considers the situations of attacker, defender and network environment together; the participants act in turn at each state and the interactions among the participants will cause the network state transition with certain probability; each state transition produces immediate payoff to attacker and defender, and the former (positive values) is in terms of the extent of damage he does to the network while the latter (negative values) is measured by the time of recovery; the future offensive-defensive behaviors will impact on the action choice of attacker and defender at each state. Nash Equilibrium strategy represents a stable plan of action for attacker and defender in long run, while the Social Optimal strategy is a policy to minimize the damage caused by attacker.

Assuming $S$ is the set of network system state, ranged over by $s_{\textit{i}}$, $i\in I$, $I$ is an index set; action sets of attacker and defender are $A^{a}$ and $A^{d}$ respectively, $u,v$ represent the general values,  $A^{a}(s_{\textit{i}})\subseteq A^{a}$ is the action set of attacker at $s_{\textit{i}}$, as well as $A^{d}(s_{\textit{i}})\subseteq A^{d}$ is that of defender; state transition probability is a function $\dot{p}:S\times A^{a}\times A^{d}\times S\rightarrow [0,1]$, and immediate payoff associated with each transition is a function $\dot{r}: S\times A^{a}\times A^{d}\rightarrow \mathds{R}_1\times \mathds{R}_2$, where $\mathds{R}$ is the real number set, and we use index to distinguish the first and the second element, and $\dot{r}^a: S\times A^{a}\times A^{d}\rightarrow \mathds{R}_{1}$ represents the immediate payoff of attacker, while $\dot{r}^d: S\times A^{a}\times A^{d}\rightarrow \mathds{R}_{2}$ is that of defender.

\textbf{ComModel}, a model based on PVCCS\textsubscript{R}, is used to modeling for the network security scenario depicted as above. The processes represent all possible behaviors of the participants in network system at each state. Each state is assigned with a process depicting all possible interactions currently performed among the participants. Then we establish a network state transition system based on the process transitions.

In \textbf{ComModel}, the channel set  $\mathcal{A}=\{Attc,\textit{Defd},Tell_a,$ $Tell_d\}$, $\mathbf{Label}=\mathcal{A}\cup \overline{\mathcal{A}}\cup \{\overline{Log}\}\cup \{\overline{Rec}\}$. The value set $\mathbf{Val}=A^a\cup A^d\cup T$, where $T \subseteq \mathds{R}\times\mathds{R}$. $\mathbf{Var}$ is the set of value variables.  $\mathbf{Act}$ is the union of behavior sets of the three participants ($Act^a$,$Act^d$ and $Act^n$) defined as follows:
\begin{equation*}
\begin{aligned}
\mathbf{Act}=&Act^a\cup Act^d\cup Act^n\\
Act^a=&\{\overline{Attc}(v)\mid v\in A^a \}\cup\{Tell_a(x)\mid x\in \mathbf{Var}\}\\
Act^d=&\{\overline{\textit{Defd}}(v)\mid v\in A^d \}\cup \{Tell_d(x)\mid x\in \mathbf{Var}\}\\
Act^n=&\{Attc(x)\mid x\in \mathbf{Var} \}\cup \{\textit{Defd}(x)\mid x\in \mathbf{Var} \}\\
&\cup \{\overline{Tell_a}(x)\mid x\in \mathbf{Var}\cup A^d \}\\
&\cup \{\overline{Tell_d}(x)\mid x\in \mathbf{Var}\cup A^a \}\\
&\cup \{\overline{Log}(x,y)\mid x\in A^a\cup \mathbf{Var}, y\in A^d\cup \mathbf{Var}\}.\\
&\cup \{\overline{Rec}(\dot{r}(s,u,v))\mid s\in S, u\in A^a, v\in A^d\}\\
\end{aligned}
\end{equation*}

\begin{figure}[!htpb]
\begin{center}
\includegraphics[scale=0.3]{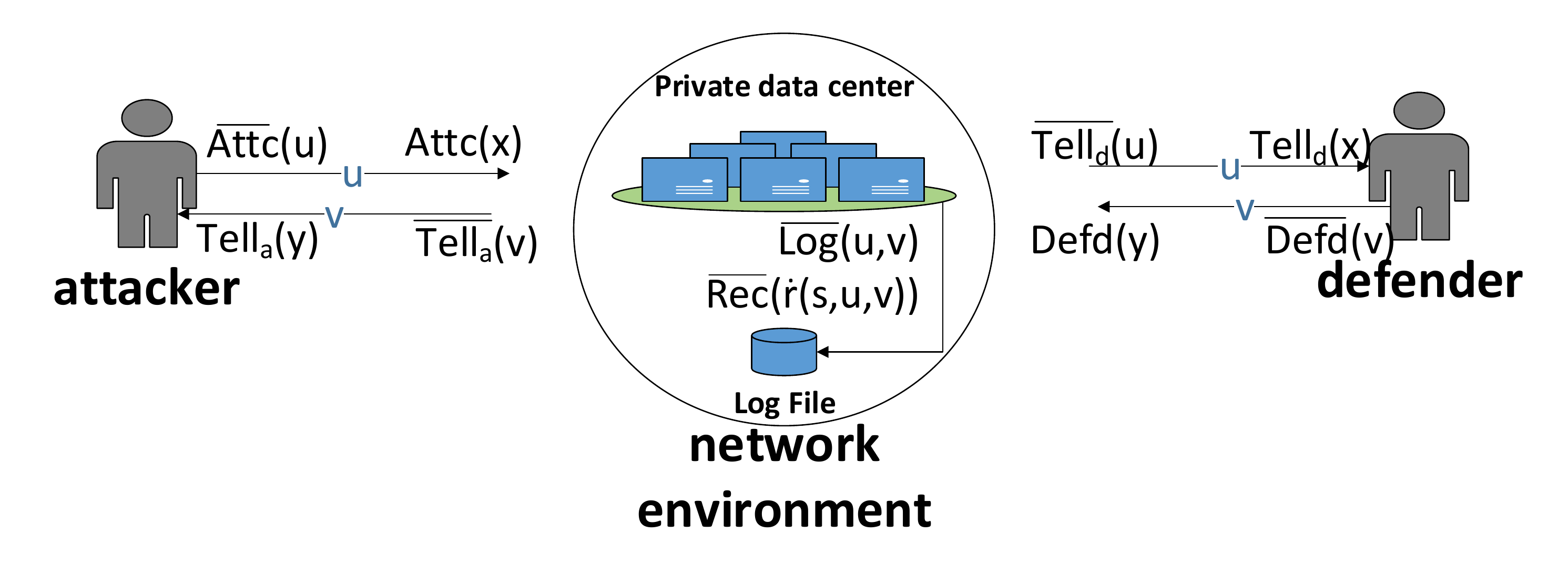}
\caption{Interactions among participants at state $s$}
\label{interactions}
\end{center}
\end{figure}
Figure \ref{interactions} shows one interaction among the participants at state $s$. $\overline{Attc}(u)$ means attacker takes attack $u$, similar to $\overline{\textit{Defd}}(v)$ for defender; $Attc(x)$ (or $\textit{Defd}(x)$) means network environment is attacked (or is defended); $\overline{Tell_d}(x)$ (or $\overline{Tell_a}(x)$) means network environment informs defender (or attacker) the action chosen by attacker (or defender); $Tell_d(x)$ (or $Tell_a(x)$) means defender (or attacker) is informed that attack (or defense) has happened; $\overline{Log}(x,y)$ means the network environment writes the values of $x$ and $y$ into a log file, where $x$ and $y$ is used to receive the values of  attack and defense respectively; $\overline{Rec}(\dot{r}(s,u,v))$ stands for the network environment records the immediate payoff to attacker and defender if they choose $u$ and $v$ at state $s$ respectively.

The processes describing all possible behaviors of the participants at state $s_{\textit{i}}$, denoted by $\textit{pA}_{i}$, $\textit{pD}_{i}$ and $\textit{pN}_{i}$, are defined as follows:
\begin{align*}
\textit{pA}_{i}\stackrel{\textit{def}}{=}&\underset{u\in A^a(s_{\textit{i}})}{\sum}\overline{Attc}(u).Tell_a(y).Nil        \\
\textit{pD}_{i}\stackrel{\textit{def}}{=}&~Tell_d(x).\underset{v\in A^d(s_{\textit{i}})}{\sum}\overline{\textit{Defd}}(v).Nil\\
\textit{pN}_{i}\stackrel{\textit{def}}{=}&   ~Attc(x).\overline{Tell_d}(x).\textit{Defd}(y).\overline{Tell_a}(y). Tr_{\textit{i}}(x,y)\\
Tr_{\textit{i}}(x,y)\stackrel{\textit{def}}{=}&\underset{u\in A^a(s_{\textit{i}})\atop v\in A^d(s_{\textit{i}})}{\sum}\overline{Log}(u,v).(if~(x=u,y=v)~ then\\
&\underset{j\in I}{\sum}[\dot{p}(s_{\textit{i}},u,v,s_{\textit{j}})] \overline{Rec}(\dot{r}(s_{\textit{i}},u,v)). (\textit{pA}_j|\textit{pD}_j|\textit{pN}_j)\\
& else~Nil)
\end{align*}

The process assigned to each state $s_{\textit{i}}$ is defined as $$G_{\textit{i}}\stackrel{\textit{def}}{=}(\textit{pA}_i|\textit{pD}_i|\textit{pN}_i) \backslash R, R=\{Attc,\textit{Defd},Tell_a,Tell_d\}$$
We get the network state transition system, \textbf{TS} for short, based on process transitions. Minimizing $\mathbf{TS}$ by shrinking probabilistic bisimilar pairs of states. We conduct a series of path contractions on $\mathbf{TS}$ and obtained a new graph named as $\mathbf{ConTS}$ without information loss as follows:
\begin{definition}
\begin{bf}ConTS\end{bf} is a tuple $(V,E,L)$
\begin{itemize}
\item $V=\{G_{\textit{i}}\mid G_{\textit{i}}$ is the process we assign to state $s_{\textit{i}}\}$
\item $E=\{(G_i, G_{j})\mid$ ranged over $e_{\textit{ij}}$, if there exists a multi-transition $G_{\textit{i}}\stackrel{(\tau[1])^{4} \overline{Log}(u,v)[1]\overline {Rec}(\dot{r}(s_{\textit{i}},u,v))[\dot{p}(s_{\textit{i}},u,v,s_{\textit{j}})]} {\Longrightarrow}G_j\}$
\item $L(e_{\textit{ij}})=\{(L_{\textit{Act}}(e_{\textit{ij}}), L_{\textit{TranP}}(e_{\textit{ij}}),L_{\textit{WeiP}}(e_{\textit{ij}}) \mid e_{\textit{ij}} \in E\}$
\begin{itemize}
\item action pair: $L_{\textit{Act}}(e_{\textit{ij}})=(u,v)$
\begin{itemize}
\item $L_{\textit{Act}}^a(e_{\textit{ij}})=u$, $L_{\textit{Act}}^d(e_{\textit{ij}})=v$
\end{itemize}
\item transition probability: $L_{\textit{TranP}}(e_{\textit{ij}})=\dot{p} (s_{\textit{i}},u,v,s_{\textit{j}})$
\item weight pair: $L_{\textit{WeiP}}(e_{\textit{ij}})=\dot{r}(s_{\textit{i}},u,v)$
\begin{itemize}
\item $L_{\textit{WeiP}}^a(e_{\textit{ij}})=\dot{r}^a(s_{\textit{i}},u,v)$
\item  $L_{\textit{WeiP}}^d(e_{\textit{ij}})=\dot{r}^d(s_{\textit{i}},u,v)$
\end{itemize}
\end{itemize}
\end{itemize}
\end{definition}
$L_{\textit{WeiP}}^S(e_{\textit{ij}})=L_{\textit{WeiP}}^a(e_{\textit{ij}})
+|L_{\textit{WeiP}}^d(e_{\textit{ij}})|$ denotes the sum of absolute weight pair of $e_{\textit{ij}}$.
By convention, in any network security scenario, for any $e$, $e'\in E$, if $L_{\textit{WeiP}}^a(e)>L_{\textit{WeiP}}^a(e')$ then $L_{\textit{WeiP}}^d(e)<L_{\textit{WeiP}}^d(e')$.

\section{Analyzing Properties as Graph Theory Approach}
We firstly introduce the definitions of Nash Equilibrium strategy (\textit{NES}) and Social Optimal strategy (\textit{SOS}) in our model, and then we illustrate the algorithms proposed to comput \textit{NES} and \textit{SOS} respectively.
\subsection{NES and SOS}
\begin{definition}
$\forall G_{\textit{i}}\in V$, an \textsf{execution} of $G_{\textit{i}}$ in \begin{bf}{ConTS}\end{bf}, denoted by $\pi_{\textit{i}}$, is a walk (vertices and edges appearing alternately) starting from $G_{\textit{i}}$ and ending with a cycle, on which every vertex's out-degree is 1.
\end{definition}
According to the definition of execution, $\pi_{\textit{i}}$ is in the form of $G_{\textit{i}}e_{\textit{ij}}G_{\textit{j}}...(G_{\textit{k}}...G_{\textit{l}}e_{\textit{lk}}G_{\textit{k}})$ which is ended by a cycle starting with $G_{\textit{k}}$, where $G_{\textit{i}}$ and $G_{\textit{k}}$ may be the same node. $\pi_{i}$ can be written as $\pi_{i}^e$ if $e$ is the first edge of $\pi_{i}$;
$\pi_{\textit{i}}[j]$ denotes the subsequence of $\pi_{\textit{i}}$ starting from $G_{\textit{j}}$, where $G_{\textit{j}}$ is a vertex on $\pi_{\textit{i}}$.
\begin{definition}
The \textsf{payoff} to attacker and defender on execution $\pi_{\textit{i}}$, denoted by $PF^a(\pi_{\textit{i}})$ and $PF^d(\pi_{\textit{i}})$ respectively, are defined as follows:
\begin{align*}
PF^a(\pi_{\textit{i}})&=L^a_{\textit{WeiP}}(e_{\textit{ij}}) +\beta\cdot L_{\textit{TranP}}(e_{\textit{ij}})\cdot PF^a(\pi_{\textit{i}}[j])\\
PF^d(\pi_{\textit{i}})&=L^d_{\textit{WeiP}}(e_{\textit{ij}}) +\beta\cdot L_{\textit{TranP}}(e_{\textit{ij}})\cdot PF^d(\pi_{\textit{i}}[j])
\end{align*}
where $\beta\in(0,1)$ is a discount factor. The sum of absolute payoff on $\pi_{\textit{i}}$ of attacker and defender is denoted as $PF^{\textit{S}}(\pi_{\textit{i}})$, and  $PF^{\textit{S}}(\pi_{\textit{i}})=PF^a(\pi_{\textit{i}})+ |PF^d(\pi_{\textit{i}})|$.
\end{definition}
\begin{Theorem}
\label{payoffconverged}
$\forall G_{\textit{i}}\in V$, $\pi_{\textit{i}}$ is an execution of $G_{\textit{i}}$, $PF^a(\pi_{\textit{i}})$ and  $PF^d(\pi_{\textit{i}})$ are converged.
\end{Theorem}
\begin{proof}
Based on the definition of payoff on an execution of $G_{\textit{i}}$ and limiting laws, we show the proof details for $PF^a(\pi_{\textit{i}})$ in Appendix. The proof for $PF^d(\pi_{\textit{i}})$ is similar.
\end{proof}
Nash Equilibrium Execution and Social Optimal Execution are defined coinductively \cite{davide07} as follows:
\begin{definition}
$\pi_{\textit{i}}$ is \textsf{Nash Equilibrium Execution (NEE)} of $G_{\textit{i}}$ if it satisfies:
\begin{align*}
PF^a(\pi_{\textit{i}})&=\underset{e_{\textit{ij}}\in E'(G_{\textit{i}})}{\max}\{L_{\textit{WeiP}}^a({e_{\textit{ij}}})+ \beta\cdot L_{\textit{TranP}}(e_{\textit{ij}})\cdot PF^a(\pi_{\textit{j}})\} \\
PF^d(\pi_{\textit{i}})&=\underset{e_{\textit{ij}}\in E^a_{e}(G_{\textit{i}})}{\max}\{L_{\textit{WeiP}}^d({e_{\textit{ij}}})+ \beta\cdot L_{\textit{TranP}}(e_{\textit{ij}})\cdot PF^d(\pi_{\textit{j}}) \}
\end{align*}
where $\pi_{\textit{j}}$ is \textit{NEE} of $G_{\textit{j}}$, $e$ is the first edge of $\pi_{\textit{i}}$, $E^a_{e}(G_{\textit{i}})=\{e'\in E(G_{\textit{i}})\mid L_{\textit{Act}}^a(e')=L_{\textit{Act}}^a(e)\}$ including $e$, and  $E'(G_{\textit{i}})=\{\arg\underset{e'\in E^a_{e''}(G_{\textit{i}})}{\max}\{L_{\textit{WeiP}}^d(e')+ \beta\cdot L_{\textit{TranP}}(e')\cdot PF^d(\pi_{\textit{j}})\}, \forall e''\in E(G_{\textit{i}})\}$.
\end{definition}

\begin{definition}
$\pi_{\textit{i}}$ is \textsf{Social Optimal Execution (SOE)} of $G_{\textit{i}}$, if it satisfies:
$$PF^{\textit{S}}(\pi_{\textit{i}})=\underset{e_{\textit{ij}}\in E(G_{\textit{i}})}{\min}\{L_{\textit{WeiP}}^S({e_{\textit{ij}}})+ \beta\cdot L_{\textit{TranP}}(e_{\textit{ij}})\cdot PF^S(\pi_{\textit{j}}) \}$$
where $\pi_{\textit{j}}$ is \textit{SOE} of $G_{\textit{j}}$
\end{definition}

\begin{definition}
\textsf{Strategy} is a sequence consisting of action pair (one from attacker and one from defender) at each state.
\end{definition}

\begin{definition}
\textsf{Nash Equilibrium Strategy} (\textit{NES}) is a strategy of which every $G_{\textit{i}}$'s execution based on is NEE of $G_{\textit{i}}$.
\end{definition}

\begin{definition}
\textsf{Social Optimal Strategy} (\textit{SOS}) is a strategy of which every $G_{\textit{i}}$'s execution based on is SOE of $G_{\textit{i}}$.
\end{definition}

\subsection{Algorithms}
The way to compute $\textit{NES}$ (or $\textit{SOS}$) in \textbf{ConTS} is to find a spanning subgraph of \textbf{ConTS} satisfying following conditions:\\
$\mathds{A}$. Each vertex's outdegree is 1;\\
$\mathds{B}$. Each vertex's execution in this subgraph is its $\textit{NEE}$ (or $\textit{SOE}$).

For backward inductive analysis, we firstly find SCC of \textbf{ConTS} based on Tarjan's algorithm \cite{reinhard} and construct \textbf{Abstraction} (\textbf{Abs} for short) by viewing each SCC as one vertex. $\mathds{V}$(\textbf{Abs}) denotes the vertex set of \textbf{Abs} ranged over by $D$. \textbf{Abs} is a DAG, and we rename $D$ with $\textit{Leave}$ if its out-degree is 0, else with \textit{Non-Leave}. By convention, $\forall D\in \mathds{V}$(\textbf{Abs}), $\mathcal{V}(D)=\{G_{\textit{i}}\in V\mid G_{\textit{i}}$ belongs to the SCC represented by $D\}$.
\begin{definition}
$\forall D\in \mathds{V}(\begin{bf}{Abs}\end{bf})$, the \textsf{priority} of $D$, denoted by \textit{prior}(D), is defined inductively:\\
(1) $prior(D)=n$, if $D$ is a $\textit{Leave}$, and $n$ is the size of $\mathds{V}$(\begin{bf}{Abs}\end{bf}).\\
(2) $prior(D)=\min\{prior(D')-1\mid $ $D'$ is any direct successor of $D$ in \begin{bf}Abs\end{bf}$\}$
\end{definition}

\begin{definition}
$D$ \textsf{depends on} $D'$ if $D'$ appears in one of the paths starting from $D$ in \begin{bf}Abs\end{bf}.
\end{definition}

\begin{Theorem}
If $prior(D)<prior(D')$ then $D'$ does not depend on $D$.
\end{Theorem}
\begin{proof}
We prove it by contradiction: if $D'$ depends on $D$, then $D$ appears in one of the paths starting from $D'$ in \begin{bf}Abs\end{bf}, so we have $prior(D')=\min\{prior(D'')-1\mid $ $D''$ is any direct successor of $D'$ in \begin{bf}Abs\end{bf}$\} < prior(D)$, contradiction.
\end{proof}

If $D$ does not depend on $D'$, then computing \textit{NES/SOS} of $D'$ has no impaction on computing \textit{NES/SOS} of $D$. To find \textit{NES/SOS} of $D$ is to find \textit{NEE/SOE} of all $G_{\textit{i}}\in \mathcal{V}(D)$.

The algorithms for computing $\textit{NES}$ and $\textit{SOS}$, denoted as $\mathbf{AlgNES}()$ and $\mathbf{AlgSOS}()$ respectively, are both based on backward induction. The framework of $\mathbf{AlgNES}()$ is as follows:\\
(1) Compute priority of each vertex $D$ in \textbf{Abs};\\
(2) Compute $\textit{NES}$ for $\textit{Leave}$ firstly, then compute backward inductively for $\textit{Non-Leave}$.\\
The framework of $\mathbf{AlgSOS}()$ is similar.

Pseudo code of $\mathbf{AlgNES}()$ is shown in Algorithm \ref{algorithm}.

\begin{algorithm}[h]
\scriptsize
 \KwData{\textbf{Abs}}
 \KwResult{$\textit{NES}$ of \textbf{Abs}}
 $\textit{NES}(\textbf{Abs})\leftarrow \emptyset$\;
 \For{$D\in \mathds{V}$(\begin{bf}{Abs}\end{bf})}{
 $prior(D)\leftarrow$ ComputePrior($D$)\;
 }
 List $\mathcal{L}\leftarrow$ list of $\mathds{V}(\textbf{Abs})$ in descending order on priority\;
 pointer $p\leftarrow \mathcal{L}$\;
 \While{$p$ is not the tail of $\mathcal{L}$}{
 $D\leftarrow p.data$\;
 \While{prior(D) is the highest in $\mathcal{L}$}{
 $\textit{NES}(\textbf{Abs})\cup\leftarrow \begin{bf}NESinLeave\end{bf}(D)$\;
 $p\leftarrow p.next$\;
 $D\leftarrow p.data$\;
 }
 $\textit{NES}(\textbf{Abs})\cup\leftarrow \begin{bf}NESinNonLeave\end{bf}(D)$\;
 $p\leftarrow p.next$\;
 $D\leftarrow p.data$\;
 }
 \caption{Pseudo code of $\mathbf{AlgNES}$()}
 \label{algorithm}
\end{algorithm}

\textbf{NES/SOS for Leave}\\
The key point of computing $\textit{NES}$ (or $\textit{SOS}$) for $\textit{Leave}$ $D$ is to find a cycle in $D$ satisfying conditions $\mathds{A}$ and $\mathds{B}$ as above.

\textbf{NES in Leave}: The method of finding $\textit{NES}$ for $\textit{Leave}$ $D$ is a value iteration method, denoted as \textbf{NESinLeave}($D$). The value function is \textbf{BackInd}($G_{\textit{i}}$) which returns some edge $e$ of $G_{\textit{i}}$ and \textbf{RefN}($G_{\textit{i}}$) is used to refresh the value of the weight pair for each edge of $G_{\textit{i}}$, $\forall G_{\textit{i}}\in \mathcal{V}(D)$.

As the narrative convenience, we introduce some auxiliary symbols: $\forall e\in E(G_{\textit{i}})$, the weight pair initializes with $L_{0}(e)=L_{\textit{WeiP}}(e)$, and $L_{\textit{n}}(e)=(L_{\textit{n}}^a(e), L_{\textit{n}}^d(e))$ is used to keep the new weight pair of $e$ obtained by \textbf{RefN}($G_{\textit{i}}$) on the nth iteration; $\forall G_{\textit{i}}\in \mathcal{V}(D)$, $Pp_{\textit{n}}(G_{i})=(Pp_{\textit{n}}^a(G_{i}), Pp_{\textit{n}}^d(G_{i}))$, initialized with $Pp_0(G_{i})=(0,0)$, is used to keep $L_{\textit{n}}(e)$, where $e$ is the result of \textbf{BackInd}($G_{\textit{i}}$) on the nth iteration. The iterative process will be continued until $\forall G_{\textit{i}}\in \mathcal{V}(D)$, $Pp_{\textit{n}}(G_{i})=Pp_{\textit{n+1}}(G_{i})$.

The framework of \textbf{NESinLeave}($D$) is as follows:\\
(1) Value iteration initializes with \textbf{BackInd}($G_{\textit{i}}$), where for each $G_{\textit{i}}\in \mathcal{V}(D)$, the weight pair of $e\in E(G_{\textit{i}})$ is $L_0(e)$. Assuming $e$ is the result obtained by \textbf{BackInd}($G_{\textit{i}}$), then $Pp_0(G_{\textit{i}})=L_0(e)$;\\
(2) Loops through the method \textbf{RefN}($G_{\textit{i}}$) and \textbf{BackInd}($G_{\textit{i}}$) by order until $\forall G_{\textit{i}}$, $Pp_{\textit{n+1}}(G_{\textit{i}}) =Pp_{\textit{n}}(G_{\textit{i}})$;\\
(3) $\forall G_{\textit{i}}$, execute \textbf{BackInd}($G_{\textit{i}}$). The cycle obtained is what we want.

Rules of method \textbf{BackInd}$(G_{\textit{i}})$ on the nth iteration,  $n\geq0$:\\
(1) Let $E'=E(G_{\textit{i}})$;\\
(2) If $\exists e_1, e_2\in E'$ satisfying $L_{\textit{Act}}^a(e_1)=L_{\textit{Act}}^a(e_2)$, refresh $E'$ by filtering the edge $e\neq\arg\underset{e\in \{e_1,e_2\}}{\max} L_{n}^d(e)$;\\
(3) Refresh $E'$ by keeping edge $e=\arg\underset{e\in E'}{\max} L_{n}^a(e)$\\
(4) Return $e$.

Rules in method \textbf{RefN}$(G_{\textit{i}})$ on the (n+1)th iteration, $n\geq0$:\\
(1) $\forall e\in E(G_{\textit{i}})$, compute its  $L_{\textit{n+1}}(e)$ componentwise by following formula:
$$L_{\textit{n+1}}(e_{\textit{ij}})=L_{\textit{WeiP}}(e_{\textit{ij}})+\beta \cdot L_{\textit{TranP}}(e_{\textit{ij}})\cdot Pp_{\textit{n}}(G_{\textit{j}})$$
(2) Keep $L_{\textit{n+1}}(e_{\textit{ij}})$, $\forall e_{\textit{ij}}\in E(G_{\textit{i}})$;

Pseudo code of \textbf{NESinLeave}(), \textbf{BackInd}() and \textbf{RefN}() are shown in Algorithm \ref{nesinleave}, \ref{backwardinduction} and \ref{refreshweip} respectively.
\begin{algorithm}[!h]
\scriptsize
 \KwData{$\textit{Leave}$ of \textbf{Abs}: $D$}
 \KwResult{$\textit{NES}$ of $D$ }
 Label $G_{\textit{i}}\in \mathcal{V}(D)$ with \begin{bf}NonConducted\end{bf}\;
 $\textit{NES}(D)\leftarrow \emptyset$\;
 \While{$\exists G_{\textit{i}}$ is \begin{bf}NonConducted\end{bf}}{
 $e\leftarrow \begin{bf}BackInd\end{bf}(G_{\textit{i}})$\;
 $Pp_{0}(G_{\textit{i}})\leftarrow L_{0}(e)$\;
 Label $G_{\textit{i}}$ with \begin{bf}Conducted\end{bf}\;
 }
 \While{$\exists G_{\textit{i}}$ $Pp_{\textit{n}}(G_{\textit{i}})\neq Pp_{\textit{n+1}}(G_{\textit{i}})$ componentwise}{
 \begin{bf}RefN\end{bf}($G_{\textit{i}}$) \;
 $e\leftarrow \begin{bf}BackInd\end{bf}(G_{\textit{i}})$\;
 $Pp_{\textit{n+1}}(G_{\textit{i}})\leftarrow L_{\textit{n+1}}(e)$\;
 }
 $\textit{NES}(D)\leftarrow \{e|e\leftarrow \begin{bf}BackInd\end{bf}(G_{\textit{i}}), G_{\textit{i}}\in \mathcal{V}(D)\}$\;
 \caption{Pseudo code of $\mathbf{NESinLeave}$()}
 \label{nesinleave}
\end{algorithm}

\begin{algorithm}[!h]
\scriptsize
 \KwData{$G_{\textit{i}}\in \mathcal{V}(D)$}
 \KwResult{edge $e\in E(G_{\textit{i}})$}
 create $E'\leftarrow E(G_{\textit{i}})$\;
 \While{$\forall e_{1}, e_{2}\in E'$ with $L_{\textit{Act}}^a(e_{1})=L_{\textit{Act}}^a(e_{2})$}{
 $E'\leftarrow E'\backslash \{e\mid e\neq\underset{e\in \{e_{1},e_{2}\}}{\arg\max}\{L_{\textit{n}}^d(e)\}\}$\;
 }
 $e\leftarrow\underset{e\in E'}{\arg\max}\{L_{\textit{n}}^a(e)\}$\;
 return $e$\;
 \caption{Pseudo code of $\mathbf{BackInd}$()}
 \label{backwardinduction}
\end{algorithm}

\begin{algorithm}[!h]
\scriptsize
 \KwData{$G_{\textit{i}}\in \mathcal{V}(D)$}
 \KwResult{$L_{\textit{n+1}}(e_{\textit{ij}})$, $\forall e_{\textit{ij}}\in E(G_{\textit{i}})$}
 Label all $e_{\textit{ij}}\in E(G_{\textit{i}})$ with \begin{bf}NonRef\end{bf}\;
 \While{$\exists e_{\textit{ij}}$ is \begin{bf}NonRef\end{bf}}{
 $L_{\textit{n+1}}^a(e_{\textit{ij}})\leftarrow L_{\textit{WeiP}}^a(e_{\textit{ij}})+ \beta\cdot L_{\textit{TranP}}(e_{\textit{ij}})\cdot Pp_{\textit{n}}^a(G_{\textit{j}})$\;
 $L_{\textit{n+1}}^d(e_{\textit{ij}})\leftarrow L_{\textit{WeiP}}^d(e_{\textit{ij}})+ \beta\cdot L_{\textit{TranP}}(e_{\textit{ij}})\cdot Pp_{\textit{n}}^d(G_{\textit{j}})$\;
 Label $e_{\textit{ij}}$ with \begin{bf}Ref\end{bf}\;
 }
 \caption{Pseudo code of $\mathbf{RefN}$()}
 \label{refreshweip}
\end{algorithm}

\textbf{SOS in Leave}: The method \textbf{SOSinLeave}($D$) used to find $\textit{SOS}$ for $\textit{Leave}$ $D$ is also a value iteration method. The \textit{value function} is \textbf{LocSoOp}($G_{\textit{i}}$) which returns some edge $e$ of $G_{\textit{i}}$ and \textbf{RefS}($G_{\textit{i}}$) is used to refresh the absolute sum value of the weight pair for each edge of $G_{\textit{i}}$, $\forall G_{\textit{i}}\in \mathcal{V}(D)$.

Here are some other auxiliary symbols for convenience: $\forall e\in E(G_{\textit{i}})$, its sum of absolute weight pair initializes with $L_0^{\textit{S}}(e)=L_{\textit{WeiP}}^S(e)$, and $L_{\textit{n}}^{\textit{S}}(e)$ is used to keep the new sum of absolute weight pair of $e$ obtained by \textbf{RefS}($G_{\textit{i}}$) on the nth iteration; $Ps_{\textit{n}}(G_{i})$ initialized with $Ps_0(G_{i})=0$, is used to keep $L_{\textit{n}}^{\textit{S}}(e)$, where $e$ is the result of \textbf{LocSoOp}($G_{\textit{i}}$) on the nth iteration. The iterative process will be continued until $\forall G_{\textit{i}}\in \mathcal{V}(D)$, $Ps_{\textit{n}}(G_{i})=Ps_{\textit{n+1}}(G_{i})$.

The framework of \textbf{SOSinLeave}($D$) is as follows:\\
(1) Value iteration initializes with \textbf{LocSoOp}($G_{\textit{i}}$), where for $G_{\textit{i}}\in \mathcal{V}(D)$, the sum of absolute weight pair of $e\in E(G_{\textit{i}})$ is $L_0^{\textit{S}}(e)$. Assuming the result obtained by \textbf{LocSoOp}($G_{\textit{i}}$) is $e$, then $Ps_0(G_{\textit{i}})=L_0^{\textit{S}}(e)$;\\
(2) Loops through the method \textbf{RefS}($G_{\textit{i}}$) and \textbf{LocSoOp}($G_{\textit{i}}$) by order until $\forall G_{\textit{i}}$, $Ps_{\textit{n+1}}(G_{\textit{i}}) =Ps_{\textit{n}}(G_{\textit{i}})$;\\
(3) $\forall G_{\textit{i}}$, execute \textbf{LocSoOp}($G_{\textit{i}}$). The cycle obtained is what we want.

Rules of method \textbf{LocSoOp}$(G_{\textit{i}})$ on nth iteration, $n\geq0$:\\
(1) Compare $L_n^S(e)$, $\forall e\in E(G_{\textit{i}})$;\\
(2) Return edge $e=\arg\underset{e\in E(G_{\textit{i}})}{\min} \{L_n^S(e)\}$.

Rules of method \textbf{RefS}$(G_{\textit{i}})$ on (n+1)th iteration, $n\geq0$: \\
(1) $\forall e_{\textit{ij}}\in E(G_{\textit{i}})$, compute its  $L_{\textit{n+1}}^S(e_{\textit{ij}})$ by following formula:
$$L_{\textit{n+1}}^S(e_{\textit{ij}})=L_{\textit{WeiP}}^{S}(e_{\textit{ij}})+ \beta\cdot L_{\textit{TranP}}(e_{\textit{ij}})\cdot  Ps_{\textit{n}}(G_{\textit{j}})$$
(2) Keep $L_{\textit{n+1}}^S(e_{\textit{ij}})$, $\forall e\in E(G_{\textit{i}})$;

Pseudo code of \textbf{SOSinLeave}(), \textbf{LocSoOp}() and \textbf{RefS}() are given in Algorithm \ref{sosinleave}, \ref{localsocialoptimal} and \ref{refreshsum} respectively.
\begin{algorithm}[!h]
\scriptsize
 \KwData{$\textit{Leave}$ of \textbf{Abs}: $D$}
 \KwResult{$\textit{SOS}$ of $D$ }
 Label $G_{\textit{i}}\in \mathcal{V}(D)$ with \begin{bf}NonConducted\end{bf}\;
 $\textit{SOS}(D)\leftarrow \emptyset$\;
 \While{$\exists G_{\textit{i}}$ is \begin{bf}NonConducted\end{bf}}{
 $e\leftarrow$ \begin{bf}LocSoOp\end{bf}($G_{\textit{i}}$)\;
 $Ps_{0}(G_{\textit{i}})\leftarrow L_{0}^S(e)$\;
 Label $G_{\textit{i}}$ with \begin{bf}Conducted\end{bf}\;
 }
 \While{$\exists G_{\textit{i}}$ $Ps_{n}(G_{\textit{i}})\neq Ps_{n+1}(G_{\textit{i}})$}{
 \begin{bf}RefS\end{bf}($G_{\textit{i}}$) \;
 $e\leftarrow$ \begin{bf}LocSoOp\end{bf}($G_{\textit{i}}$)\;
 $Ps_{n+1}(G_{\textit{i}})\leftarrow L_{n+1}^S(e)$\;
 }
 $\textit{SOS}(D)\leftarrow \{e|e\leftarrow \begin{bf}LocSoOp\end{bf}(G_{\textit{i}}), G_{\textit{i}}\in \mathcal{V}(D)\}$\;
 \caption{Pseudo code of $\mathbf{SOSinLeave}$()}
 \label{sosinleave}
\end{algorithm}
\begin{algorithm}[!h]
\scriptsize
 \KwData{$G_{\textit{i}}\in \mathcal{V}(D)$}
 \KwResult{edge $e\in E(G_{\textit{i}})$}
 \While{$\exists e\in E(G_{\textit{i}})$ is not compared}{
 $e'\leftarrow\arg\underset{e\in E}{\min}\{L_{\textit{n}}^S(e)\}$\;}
 return $e'$\;
 \caption{Pseudo code of $\mathbf{LocSoOp}$()}
 \label{localsocialoptimal}
\end{algorithm}
\begin{algorithm}[!h]
\scriptsize
 \KwData{$G_{\textit{i}}\in \mathcal{V}(D)$}
 \KwResult{$L_{\textit{n+1}}^S(e_{\textit{ij}}), \forall e_{\textit{ij}}\in E(G_{\textit{i}})$}
 Label all $e_{\textit{ij}}\in G_{\textit{i}}$ with \begin{bf}NonRef\end{bf}\;
 \While{$\exists e_{\textit{ij}}$ is \begin{bf}NonRef\end{bf}}{
 $L_{\textit{n+1}}^S(e_{\textit{ij}})\leftarrow L_{\textit{WeiP}}^{S}(e_{\textit{ij}})+ \beta\cdot L_{\textit{TranP}}(e_{\textit{ij}})\cdot Ps_{\textit{n}}(G_{\textit{j}})$\;
 Label $e_{\textit{ij}}$ with \begin{bf}Ref\end{bf}\;
 }
 \caption{Pseudo code of $\mathbf{RefS}$()}
 \label{refreshsum}
\end{algorithm}

\textbf{NES/SOS for Non-Leave}\\
\textbf{NES of Non-Leave}: For $\textit{Non-Leave}$ vertex $D$ in \textbf{Abs}, the method of computing its NES is \textbf{NESinNonLeave}($D$) and its framework is as follows:\\
(1) if the size of $\mathcal{V}(D)$ is more than 1, we will pre-process $D$ with method \textbf{PrePro}($D$) firstly, then get its $\textit{NES}$ by \textbf{NESinLeave}($D$);\\
(2) if $\mathcal{V}(D)=\{G_{\textit{i}}\}$ for some $G_{\textit{i}}\in V$, then the $\textit{NES}$ of $D$ is the result obtained from \textbf{BackInd}($G_{\textit{i}}$) directly.

Rules in method \textbf{PrePro}($D$) are as follows:\\
(1) $D'$ is one direct successor of $D$ in \textbf{Abs}, and if the edge $e$ connecting $D$ and $D'$ is contributed by the connection between $G_{\textit{i}}\in \mathcal{V}(D)$ and $G_{\textit{j}}\in \mathcal{V}(D')$, then $L_{0}(e_{\textit{ij}})=L_{\textit{WeiP}}(e_{\textit{ij}})+\beta\cdot L_{\textit{TranP}}(e_{\textit{ij}})\cdot PF(\pi_{\textit{j}})$ componentwise, where $\pi_{\textit{j}}$ is the nash equilibrium execution of $G_{\textit{j}}$;\\
(2) Change $e$ to be the self-loop edge of $G_{\textit{i}}$.

Pseudo code of \textbf{NESinNonLeave}() and \textbf{PrePro}() are shown in Algorithm \ref{nesinnonleave} and Algorithm \ref{preprocess} respectively.
\begin{algorithm}[!h]
\scriptsize
 \KwData{$\textit{Non-Leave}$: $D$}
 \KwResult{$\textit{NES}$ of $D$}
 $\textit{NES}(D)\leftarrow \emptyset$\;
 \eIf{the size of $\mathcal{V}(D)$ is bigger than 1}{
 $D'\leftarrow$ \begin{bf}PrePro($D$)\end{bf}\;
 $\textit{NES}(D)\leftarrow $\begin{bf}NESinLeave($D'$)\end{bf}\;
 }
 {
 $\textit{NES}(D)\leftarrow $\begin{bf}BackInd\end{bf}($G_{\textit{i}}$), if $\mathcal{V}(D)=\{G_{\textit{i}}\}$\;
 }
 \caption{Pseudo code of $\mathbf{NESinNonLeave}()$}
 \label{nesinnonleave}
\end{algorithm}

\begin{algorithm}[!h]
\scriptsize
 \KwData{$\textit{Non-Leave}$: $D$}
 \KwResult{new $D'$}
 $E'\leftarrow E(G_{\textit{i}}),G_{\textit{i}}\in \mathcal{\mathcal{V}}(D)$\;
 \While{$\exists e_{\textit{ij}}\in E'$ with endpoint $G_{\textit{j}}\notin \mathcal{V}(D)$}{
 $L_{0}^a(e_{\textit{ij}})\leftarrow L_{\textit{WeiP}}^a(e_{\textit{ij}})+ \beta\cdot L_{\textit{TranP}}(e_{\textit{ij}})\cdot Pp^a(G_{\textit{j}})$\;
 $L_{0}^d(e_{\textit{ij}})\leftarrow L_{\textit{WeiP}}^d(e_{\textit{ij}})+ \beta\cdot L_{\textit{TranP}}(e_{\textit{ij}})\cdot Pp^d(G_{\textit{j}})$\;
 Change $e_{\textit{ij}}$ to be self-loop edge of $G_{\textit{i}}$\;
 }
 $D'\leftarrow (\mathcal{V}(D), E')$\;
 Return $D'$\;
 \caption{Pseudo code of $\mathbf{PrePro}()$}
 \label{preprocess}
\end{algorithm}
\textbf{SOS of Non-Leave}: The method \textbf{SOSinNonLeave}($D$) computing $\textit{SOS}$ for $\textit{Non-Leave}$ $D$ is identical to \textbf{NESinNonLeave}($D$) except for the preprocessing method \textbf{PreProS}($D$). The computing steps of \textbf{PreProS}($D$) are as follows: \\
(1) $D'$ is one direct successor of $D$ in \textbf{Abs}, if the edge $e$ connecting $D$ and $D'$ is contributed by connection between $G_{\textit{i}}\in \mathcal{V}(D)$ and $G_{\textit{j}}\in \mathcal{V}(D')$, then $L_{0}^S(e_{\textit{ij}})=L_{\textit{WeiP}}^{S}(e_{ij})+\beta\cdot L_{\textit{TranP}}(e_{\textit{ij}})\cdot PF^{S}(\pi_{\textit{j}})$, where $\pi_{\textit{j}}$ is social optimal execution of $G_{\textit{j}}$;\\
(2) Change $e$ to be self-loop edge of $G_{\textit{i}}$.

Pseudo code of \textbf{SOSinNonLeave}() and \textbf{PreProS}() is shown in Algorithm \ref{sosinnonleave} and Algorithm \ref{preprocessforSOS} respectively in Appendix.
\begin{algorithm}[!h]
\scriptsize
 \KwData{$\textit{Non-Leave}$: $D$}
 \KwResult{$\textit{SOS}$ of $D$}
 $\textit{SOS}(D)\leftarrow \emptyset$\;
 \eIf{the size of $\mathcal{V}(D)$ is bigger than 1}{
 $D'\leftarrow$ \begin{bf}PreProS($D$)\end{bf}\;
 $\textit{SOS}(D)\leftarrow $\begin{bf}SOSinLeave($D'$)\end{bf}\;
 }
 {
 $\textit{SOS}(D)\leftarrow $\begin{bf}LocSoOp\end{bf}($G_{\textit{i}}$), if $\mathcal{V}(D)=\{G_{\textit{i}}\}$\;
 }
 \caption{Pseudo code of $\mathbf{SOSinNonLeave}()$}
 \label{sosinnonleave}
\end{algorithm}
\begin{algorithm}[!h]
\scriptsize
 \KwData{$\textit{Non-Leave}$: $D$}
 \KwResult{new $D'$}
 $E'\leftarrow E(G_{\textit{i}}),G_{\textit{i}}\in \mathcal{V}(D)$\;
 \While{$\exists e_{\textit{ij}}\in E'$ with endpoint $G_{\textit{j}}\notin \mathcal{V}(D)$}{
 $L_{0}^S(e_{\textit{ij}})\leftarrow L_{\textit{WeiP}}^{S}(e_{\textit{ij}})+ \beta\cdot L_{\textit{TranP}}(e_{\textit{ij}})\cdot Ps(G_{\textit{j}})$\;
 Change $e_{\textit{ij}}$ to be self-loop edge of $G_{\textit{i}}$\;
 }
 $D'\leftarrow (\mathcal{V}(D), E')$\;
 Return $D'$\;
 \caption{Pseudo code of $\mathbf{PreProS}()$}
 \label{preprocessforSOS}
\end{algorithm}

\subsection{Correctness of Algorithms}

\textbf{Correctness of \textbf{NESinLeave}()}\\
Inspired by a technique in dynamic programming which is called \textit{value-iteration} \cite{vander, ls}, \textbf{BackInd}($D$) is formalized as a mapping $\sigma: \mathcal{V}(D)\rightarrow \mathds{R}\times \mathds{R}$, on kth iteration, $\sigma_{k}(G_{\textit{i}})=(\sigma_{k}(G_{\textit{i}})^a, \sigma_{k}(G_{\textit{i}})^d)=Pp_{k}(G_{\textit{i}})$. \textbf{RefN}() defines a set of vertex $\{G_{\textit{i}}(\sigma_{k})\mid G_{\textit{i}}(\sigma_{k})$ denotes $G_{\textit{i}}$ with $e_{\textit{ij}}$ whose weight pair is refreshed by the rule componentwise: $L_{\textit{k+1}}(e_{\textit{ij}})=L_{\textit{WeiP}}(e_{\textit{ij}})+ \beta\cdot L_{\textit{TranP}}(e_{\textit{ij}})\cdot \sigma_{k}(G_{\textit{j}})\}$. According to the rules of \textbf{NESinLeave}($D$),  $\sigma_{k+1}(G_{\textit{i}})=Pp_{1}(G_{i}(\sigma_{k}))$ for any $G_{\textit{i}}\in \mathcal{V}(D)$. It is convenient to define the shorthand operator notation $(T\sigma)(G_{\textit{i}})=Pp_{1}(G_{i}(\sigma))$, that is $T\sigma_{k}=\sigma_{k+1}$.

\begin{Lemma}
\label{provenesterminatedpre}
For any $G_{\textit{i}}\in \mathcal{V}(D)$, we have
$$\mid \sigma_{k}(G_{\textit{i}})^a- T\sigma_{k}(G_{\textit{i}})^a\mid \leq \underset{e\in E(G_{\textit{i}})}{\max} \mid L_{k}^a(e)- L_{k+1}^a(e)\mid$$
$$\mid \sigma_{k}(G_{\textit{i}})^d- T\sigma_{k}(G_{\textit{i}})^d\mid \leq \underset{e\in E(G_{\textit{i}})}{\max} \mid L_{k}^d(e)- L_{k+1}^d(e)\mid$$
\end{Lemma}

\begin{proof}
We prove by contradiction for the first inequality in details. According to the rules in \textbf{BackInd}($G_{\textit{i}}$), we need to consider all possible results obtained by \textbf{BackInd}($G_{\textit{i}}$) on kth and (k+1)th iteration respectively. The details are shown in Appendix. The proof for the second inequality is similar.
\end{proof}

\begin{Lemma}
\label{provenesterminated}
$T$ is a contraction.
\end{Lemma}
\begin{proof}
For any real vector $\overrightarrow{x}\in \mathds{R}^J$, $J$ is an index set, let $\mid \mid \overrightarrow{x}\mid\mid_{\infty}=\max_{j}|x_{j}|$. According to Lemma \ref{provenesterminatedpre}, then we have
\begin{align*}
\mid\mid T\sigma_{k+1}^a-T\sigma_{k}^a\mid\mid_{\infty}&=\underset{G_{i}\in V}{\max}\mid T\sigma_{k+1}(G_{i})^a-T\sigma_{k}(G_{i})^a \mid\\
&\leq\underset{G_{i}\in V}{\max}\underset{e_{\textit{ij}}\in E(G_{\textit{i}})}{\max} \mid L_{k+2}^a(e_{\textit{ij}})- L_{k+1}^a(e_{\textit{ij}})\mid\\
&\leq\underset{G_{j}\in V}{\max}\beta\cdot \mid   \sigma_{k+1}(G_{\textit{j}})^a-\sigma_{k}(G_{\textit{j}})^a \mid\\
&=\beta\cdot\mid\mid \sigma_{k+1}^a-\sigma_{k}^a\mid\mid_{\infty}
\end{align*}
similar proof for $\mid\mid T\sigma_{k+1}^d-T\sigma_{k}^d\mid\mid_{\infty}\leq\beta\cdot\mid\mid \sigma_{k+1}^d-\sigma_{k}^d\mid\mid_{\infty}$. Therefore, we claim that $\exists \sigma_{*}$, satisfying $T\sigma_{*}=\sigma_{*}$.
\end{proof}

\begin{Theorem}
\label{provenesinleave}
If $D$ is a $\textit{Leave}$ of \textbf{Abs}, then the result obtained by \textbf{NESinLeave}($D$) is $\textit{NES}$ of $D$.
\end{Theorem}
\begin{proof}
We need to prove two issues: \\
1. \begin{bf}NESinLeave\end{bf}($D$) is terminated.\\
2. The execution of $G_{\textit{i}}$, $\forall G_{\textit{i}}\in \mathcal{V}(D)$, based on the result of \textbf{NESinLeave}($D$) is its nash equilibrium execution.\\
The details are shown in Appendix.
\end{proof}

\textbf{Correctness of \textbf{SOSinLeave}()}\\
The way to prove the correctness of \textbf{SOSinLeave}() is similar to that of \textbf{NESinLeave}(). We will give the outline of the proofs.

We can formalize \textbf{LocSoOp}($D$) as a mapping $\alpha': \mathcal{V}(D)\rightarrow \mathds{R}$, so on kth iteration, we have $\sigma_{k}'(G_{\textit{i}})=Ps_{k}(G_\textit{i})$. \textbf{RefS}() defines a set of vertex $\{G_{\textit{i}}(\sigma'_{k})\mid G_{\textit{i}}(\sigma'_{k})$ denotes $G_{\textit{i}}$ with $e_{\textit{ij}}$ whose sum of absolute weight pair is refreshed by the rule: $L_{\textit{k+1}}^{S}(e_{\textit{ij}})=L_{\textit{WeiP}}^{S} (e_{\textit{ij}})+ \beta\cdot L_{\textit{TranP}}(e_{\textit{ij}})\cdot \sigma'_{k}(G_{\textit{j}})\}$. According to the rules of  \textbf{SOSinLeave}($D$), for any $G_{\textit{i}}\in \mathcal{V}(D)$,  $\sigma_{k+1}'(G_{\textit{i}})=Ps_{1}(G_{i}(\sigma_{k}'))$. It is convenient to define another shorthand operator notation $(T'\sigma')(G_{\textit{i}})=Ps_{1}(G_{i}(\sigma'))$, that is $T'\sigma_{k}'=\sigma_{k+1}'$. By the same way as Lemma \ref{provenesterminatedpre} and Lemma \ref{provenesterminated}, we can prove operator $T'$ is a contraction.

\begin{Lemma}
\label{provesosterminatedpre}
For any $G_{\textit{i}}\in \mathcal{V}(D)$, we have\\
$$\mid \sigma_{k}'(G_{\textit{i}})- T'\sigma_{k}'(G_{\textit{i}})\mid \leq \underset{e\in E(G_{\textit{i}})}{\max} \mid L_{k}^S(e)- L_{k+1}^S(e)\mid$$
\end{Lemma}

\begin{proof}
The proof is similar to that of Lemma \ref{provenesterminatedpre}.
\end{proof}

\begin{Lemma}
\label{provesosterminated}
$T'$ is a contraction.
\end{Lemma}

\begin{proof}
The proof is similar to that of Lemma \ref{provenesterminated}.
\end{proof}
\begin{Theorem}
If $D$ is a $\textit{Leave}$ of \textbf{Abs}, then the result obtained by \textbf{SOSinLeave}($D$) is $\textit{SOS}$ of $D$.
\end{Theorem}
\begin{proof}
The proof is similar to that of Theorem \ref{provenesinleave}.
\end{proof}

\textbf{Correctness of AlgNES() and AlgSOS()}
\begin{Theorem}
The results obtained from \textbf{AlgNES(Abs)} and \textbf{AlgSOS(Abs)} are $\textit{NES}$ and $\textit{SOS}$ of \textbf{Abs} respectively.
\end{Theorem}
\begin{proof}
We prove the correctness of \textbf{AlgNES(Abs)} in details. Prove inductively on priority of vertex $D$ in \textbf{Abs}.\\
(1) If $D$ is a $Leave$, we need to prove the result of \textbf{NESinLeave}($D$) is $\textit{NES}$ of $D$, according to Theorem \ref{provenesinleave}, trivial;\\
(2) For $\textit{Non-Leave}$ $D$, and we assume $prior(D)=prior(D')-1$, by induction hypothesis, $D'$ has got its NES by \textbf{AlgNES}(). If $\mathcal{V}(D)=\{G_{\textit{i}}\}$ for some $G_{\textit{i}}\in V$, according to the definition of $NEE$ and rules of \textbf{BackInd}(), the proof is trivial; if the size of $\mathcal{V}(D)$ is bigger than 1, according to the theorem \ref{provenesinleave}, trivial.
\end{proof}

\section{Case study}
The details of the example we used can be found in \cite{klye}. It shows a local network connected to Internet (see Figure \ref{example}). By the assumption that the firewall is unreliable, and the operating system on the machine is insufficiently hardened, the attacker has chance to pretend as a root user in web server and steal or damage data stored in private file server and private workstation.
\begin{figure}[!htpb]
\begin{center}
\includegraphics[scale=0.32]{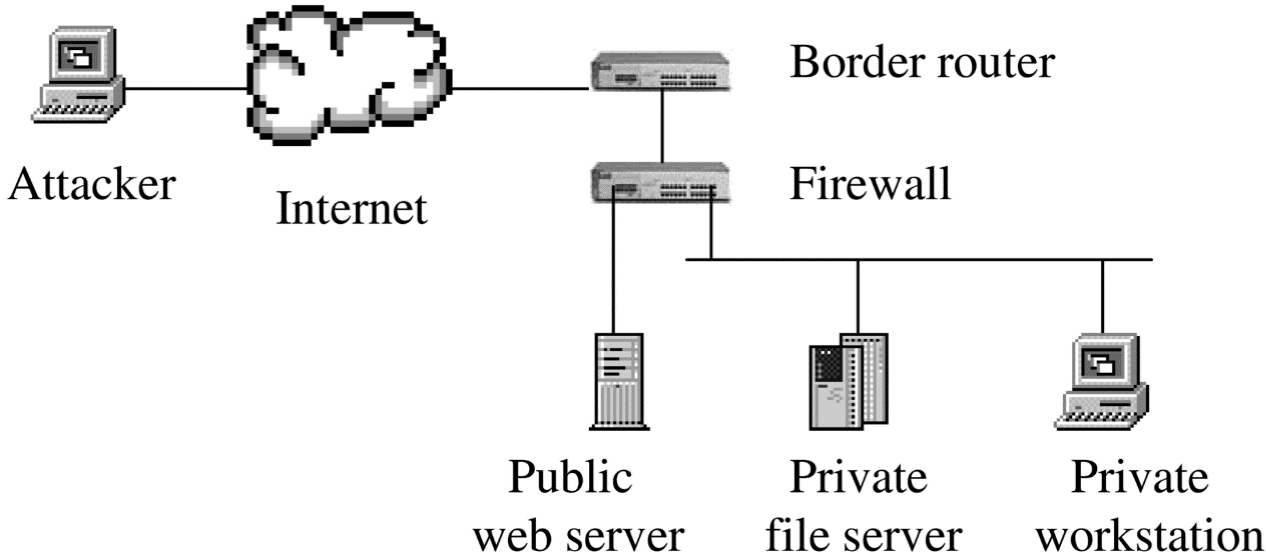}
\caption{Case study}
\label{example}
\end{center}
\end{figure}
The state set $S$ of example is shown in Table \ref{state}; $A^a$, $A^d$ is given in Table \ref{attckeraction} and Table \ref{defenderaction} respectively; for convenience, we will mostly refer to the states and actions using their symbolic number; state transition probability is shown in Table \ref{statetransition}, in which $\dot{p}(s_1,1,2,s_1)=P(1|1,1,2)$; the immediate payoff to attacker and defender at each state is shown in Table \ref{rewardcost}, in which $\dot{r}^a(s_{1},2,\cdot)=R^1(1,2,\cdot)$ and $\dot{r}^d(s_{1},2,\cdot)=R^2(1,2,\cdot)$, where $\cdot$ means any action available at current state.

\subsection{Modeling for Case study}
We modeling for state $s_1$ in \textbf{ComModel} as example, then we have $\textit{pA}_1$, $\textit{pD}_1$, $\textit{pN}_1$ as follows:
\begin{align*}
\textit{pA}_1 \stackrel{\textit{def}}{=}&  \underset{u\in A^a(s_1)}{\sum}\overline{Attc}(u).Tell_d(y).Nil \\
\textit{pD}_1  \stackrel{\textit{def}}{=}&  Tell_a(x).\underset{v\in A^d(s_1)}{\sum}\overline{\textit{Defd}}(v).Nil\\
\textit{pN}_1  \stackrel{\textit{def}}{=}&  Attc(x).\overline{Tell_a}.\textit{Defd}(y).\overline{Tell_a}.Tr_{1}(x,y)\\
Tr_{1}(x,y)  \stackrel{\textit{def}}{=}& \underset{u\in A^a(s_1)\atop v\in A^d(s_1)}{\sum}\overline{Log}(u,v).(if~(x=u,y=v)~ then\\
&  \underset{j\in I}{\sum}[\dot{p}(s_i,u,v,s_j)]\overline{Rec}(\dot{r}(s_i,u,v)). (\textit{pA}_j|\textit{pD}_j|\textit{pN}_j)\\
&  else~Nil)
\end{align*}

We find three pairs of states which are probabilistic bisimilar: $s_{13}\sim s_{15}$, $s_{14}\sim s_{16}$ and $s_{17}\sim s_{18}$. Figure \ref{ComModelcasestudy} shows the \textbf{ConTS} of case study.

\begin{figure}[htpb]
\begin{center}
\includegraphics[scale=0.32]{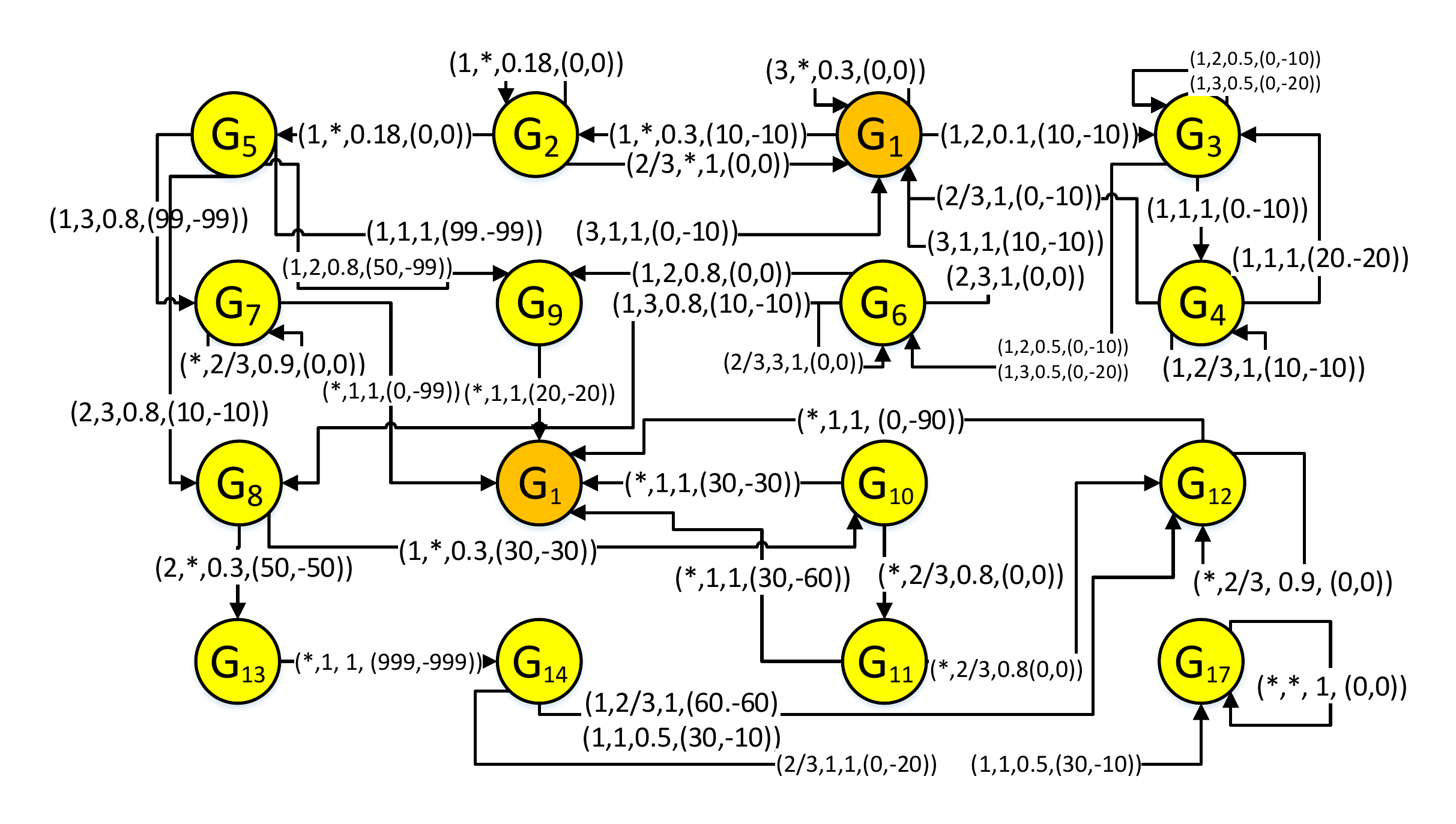}
\caption{ConTS of Example}
\label{ComModelcasestudy}
\end{center}
\end{figure}

\subsection{Analyzing NES/SOS for Case study}
We implement the algorithms using Java in Eclipse development environment on machine with 3.4GHz Inter(R) Core(TM) i72.99G RAM.
We get two Nash Equlibrium Strategies and one Social Optimal strategy for our case study, shown in Figure \ref{nash1}, \ref{nash2}, \ref{sos}  respectively.
\begin{figure}[!htb]
\begin{center}
\includegraphics[scale=0.3]{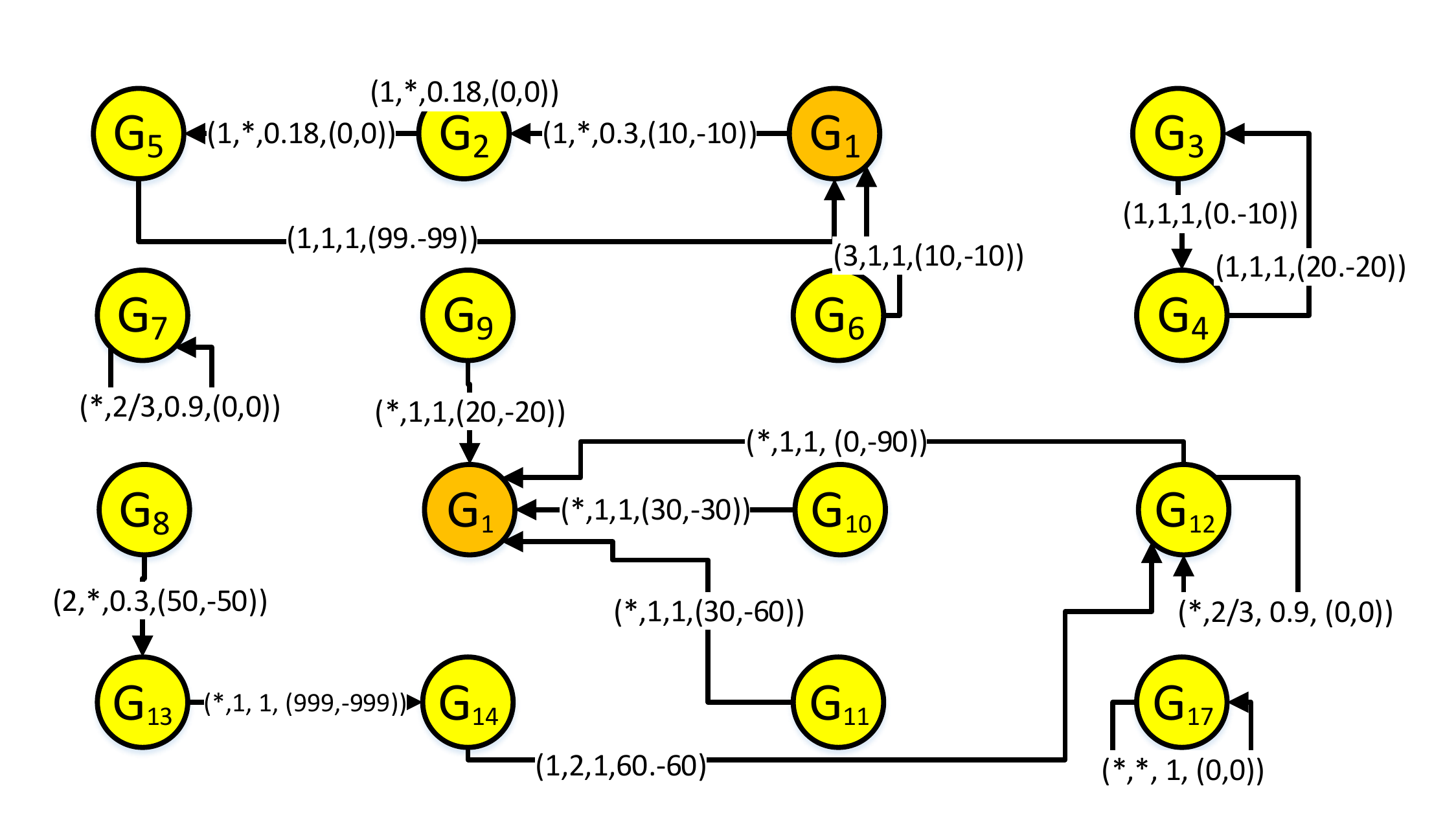}
\caption{Nash Equilibrium strategy 1}
\label{nash1}
\end{center}
\end{figure}

\begin{figure}[!htb]
\begin{center}
\includegraphics[scale=0.3]{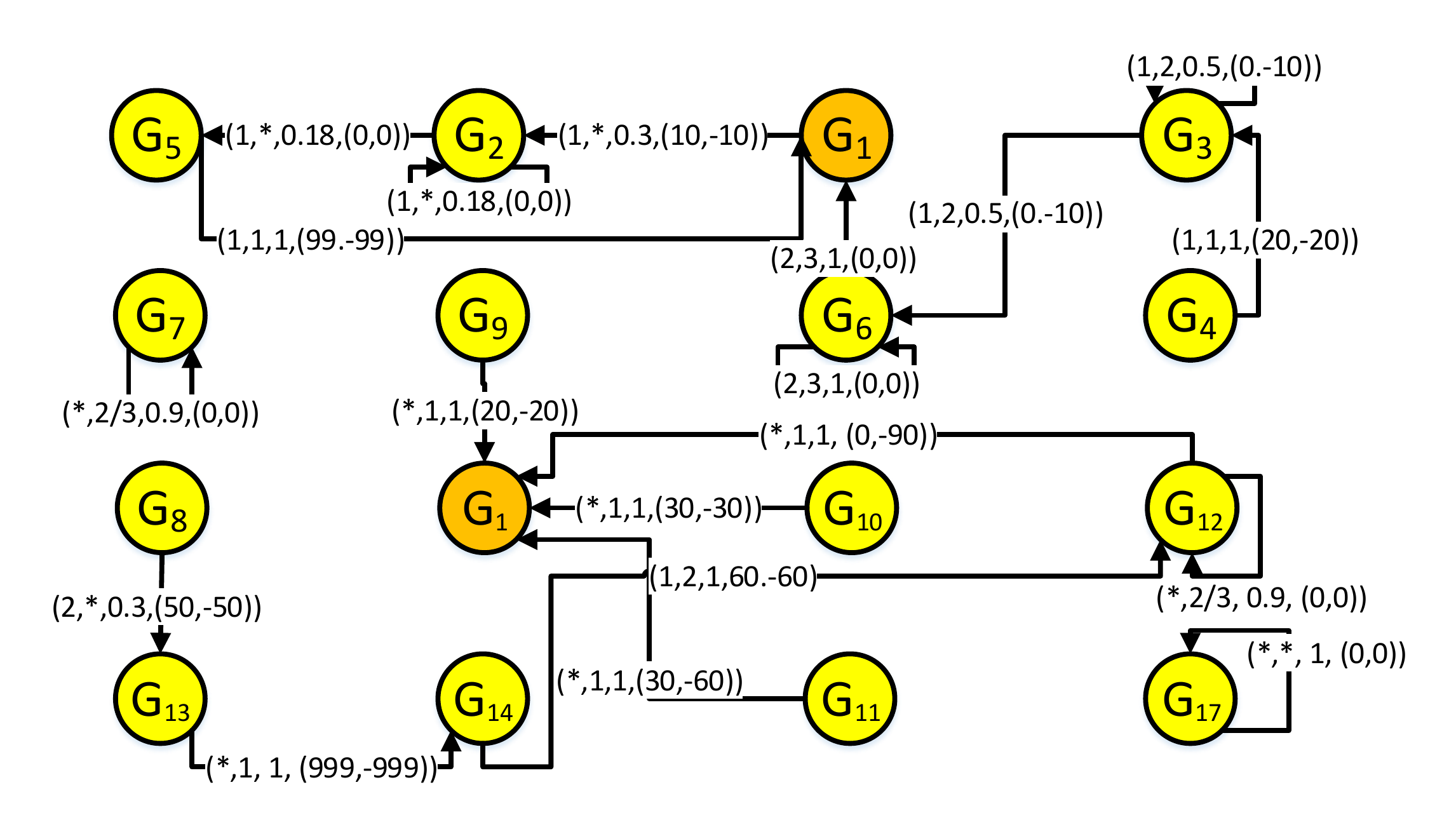}
\caption{Nash Equilibrium strategy 2}
\label{nash2}
\end{center}
\end{figure}

\begin{figure}[!htb]
\begin{center}
\includegraphics[scale=0.3]{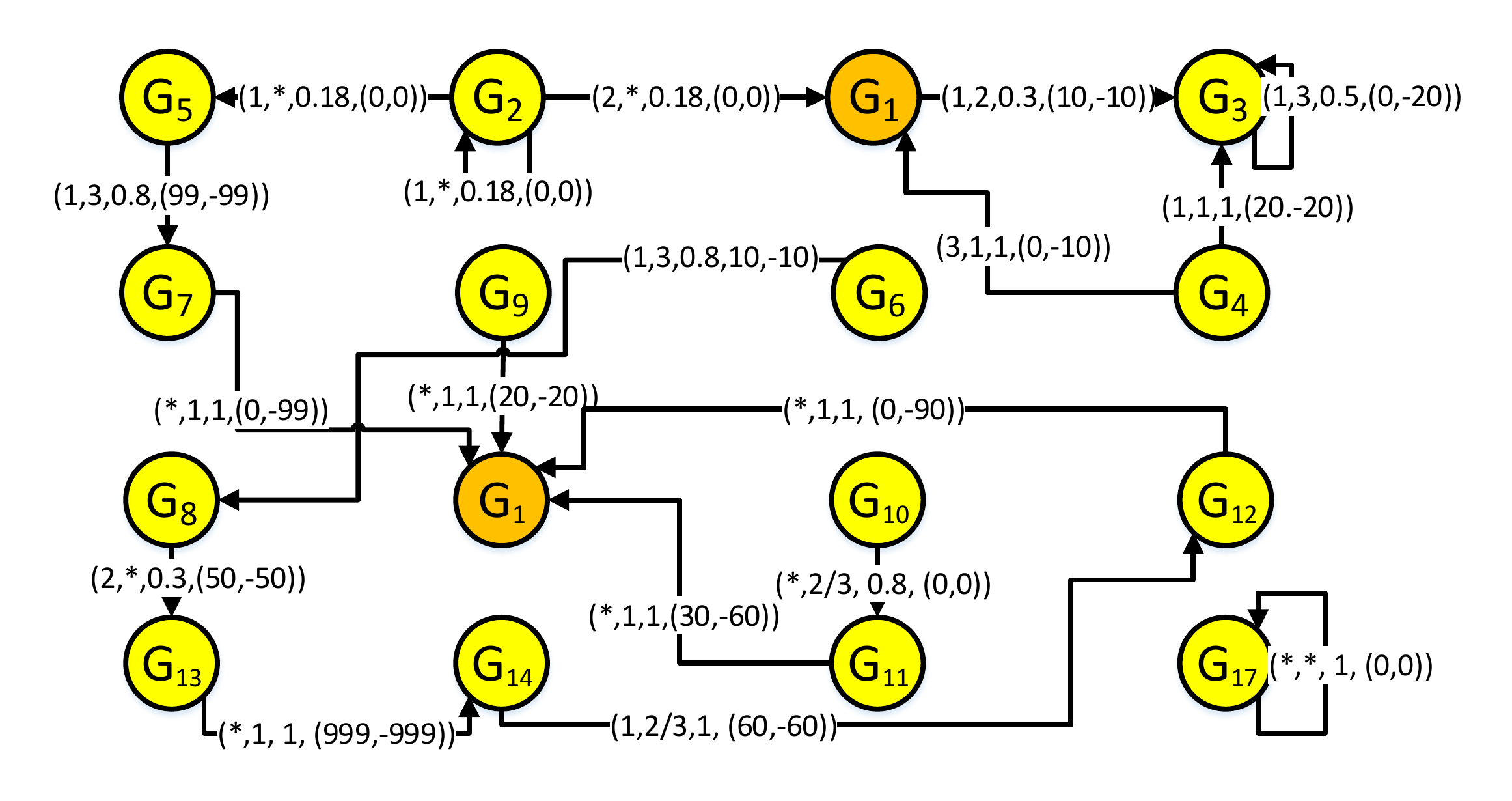}
\caption{Social Optimal strategy}
\label{sos}
\end{center}
\end{figure}

\subsection{Evaluation}
We compare our results with those obtained in \cite{klye} by game-theoretic approach: \\
(1) We filter the invalid Nash Equilibrium strategy from the results in \cite{klye}. We filter the action pair \textit{($\phi$, $Remove\_$ $\textit{Sniffer}\_$ $Detecor$)} at state $s_3$ and the action pair \textit{($Instrall\_$ $\textit{Sniffer}\_$ $Detector$, $Remove\_$Compromised\_$ $account\_ $ $restart\_$ $\textit{ftpd}$)$} at $s_6$ which obtained in the second Nash Equilibrium strategy in \cite{klye} but have no practical state transition.\\
(2) We minimize the state space by probabilistic bisimulation while   \cite{klye} focuses on the whole state set. Time consumed to compute Nash Equilibrium strategy and Social Optimal strategy for this example with our approach is shown in Table 2. Although it is incomparable with the time consumed in \cite{klye} because of evaluating on different machine models, our approach should be faster theoretically.

\begin{table}[htpd]
\centering
\scriptsize
\label{comsume}
\begin{tabular}{ccc}\hline
ComModel& Nash Equilibrium & Social Optimal\\
Creation &  strategy &  strategy \\\hline
2.8s & 3.7s & 1.4s\\\hline
\end{tabular}
\caption{Time consumed for example with our approach}
\end{table}

\section{Conclusion}
We proposed a probabilistic value-passing CCS (PVCCS) approach for modeling and analyzing a typical network security scenario with one attacker and one defender which is usually modeled by perfect and complete information game.
Extention of this method might provide uniform framework for modelling and analyzing network security scenarios which are usually modeled via different games.
We designed two algorithms for computing Nash Equilibrium strategy and Social Optimal strategy based on this PVCCS approach and on graph-theoretic methods. Advantages of these algorithms are also discussed.

\bibliographystyle{abbrv}
\bibliography{reference}

\newpage
\appendix

\section{Proofs of Theorems}
Proof of Theorem \ref{payoffconverged}
\begin{proof}
As vertex set $V$ is finite, then any infinite execution $\pi_{\textit{i}}$ of $G_{\textit{i}}$ is in form of $ G_{\textit{i}}e_{\textit{i1}}G_{\textit{1}}...(G_{\textit{k}}...G_ {\textit{(k+m)}}e_{\textit{(k+m)k}}G_{\textit{k}})$ which means ending with a cycle starting with $G_{\textit{k}}$, and $m$ is the number of vertex on this cycle except $G_{\textit{k}}$, then we have
\begin{align*}
PF^a(\pi_{\textit{i}})&=L^a_{\textit{WeiP}}(e_{\textit{i1}})+\beta\cdot PF^a(\pi_{\textit{i}}[1])\\
&=L^a_{\textit{WeiP}}(e_{\textit{i1}})+...+\beta^{k}\cdot PF^a(\pi_{\textit{i}}[k])\\
&=A+\beta^{k}\cdot PF^a(\pi_{\textit{i}}[k])\\
where~A&=L^a_{\textit{WeiP}}(e_{\textit{i1}})+...+\beta^{k-1}\cdot L^a_{\textit{WeiP}}(e_{\textit{(k-1)k}}) \\
PF^a(\pi_{\textit{i}}[k]) & =L^a_{\textit{WeiP}}(e_{\textit{k(k+1)}})+...+\beta^{m}\cdot L^a_{\textit{WeiP}}(e_{\textit{(k+m)k}})\\
&+\beta^{m+1}\cdot B+\beta^{2m+2}\cdot B+...\\
& =(B\cdot \frac{1-\beta^{(m+1)h}}{1-\beta^{(m+1)}})_{h\rightarrow \infty}\\
where~B&=L^a_{\textit{WeiP}}(e_{\textit{k(k+1)}})+...+\beta^{m}\cdot L^a_{\textit{WeiP}}(e_{\textit{(k+m)k}})\\
\underset{h\rightarrow \infty}{\lim}PF^a(\pi_{\textit{i}})&= \underset{h\rightarrow \infty}{\lim}(A+\beta^{k}(B\cdot \frac{1-\beta^{(m+1)h}}{1-\beta^{(m+1)}}))\\
&=A+\beta^{k}\cdot B\cdot \underset{h\rightarrow \infty}{\lim}(\frac{1-\beta^{(m+1)h}}{1-\beta^{(m+1)}})\\
&=A+B\cdot \frac{\beta^{k}}{1-\beta^{(m+1)}}
\end{align*}
\end{proof}
Proof of Lemma \ref{provenesterminatedpre}
\begin{proof}
Assuming without loss of generality, $\sigma_{k}(G_{\textit{i}})=(L_{k}(e_{1})^a, L_{k}(e_{1})^d)$ and $T\sigma_{k}(G_{\textit{i}})=(L_{k+1}(e_{2})^a, L_{k+1}(e_{2})^d)$, where $e_1$, $e_2\in E(G_{\textit{i}})$. Let $L_{k}^a(e_{1})=a$, $L_{k}^a(e_{2})=b$, $L_{k+1}^a(e_{1})=a'$ and $L_{k+1}^a(e_{2})=b'$, where $a$, $a'$, $b$, $b'$ are positive number.\\
\textbf{case 1}: $L_{\textit{Act}}^a(e_{1})=L_{\textit{Act}}^a(e_{2})$\\
According to the rules of \textbf{BackInd}(), we have $a<b$ and $b'<a'$. If the first inequality in lemma doesn't hold, then we have $\mid a-b'\mid >\mid a-a'\mid$ and $\mid a-b' \mid>\mid b-b'\mid$, then we get $(b'-a')(b'+a')>2a(b'-a')$ and $(a-b)(a+b)>2b'(a-b)$ which deduce $a-b>a'-b'$, contradiction.\\
\textbf{case 2}: $L_{\textit{Act}}^a(e_{1})\neq L_{\textit{Act}}^a(e_{2})$\\
Let us define two conditions:\\
Cond 1: on kth iteration, $e_{2}$ is kept by step (2) of \textbf{BackInd}(). \\
Cond 2: on (k+1)th iteration, $e_{1}$ is kept by step (2) of \textbf{BackInd}().\\
There are four subcases to be considered:\\
\textbf{$\textit{case 2.1}$}: both Cond 1 and Cond 2\\
According to the rules of \textbf{BackInd}(), we have $a>b$ and $b'>a'$. If $\mid a-b'\mid >\mid a-a'\mid$ and $\mid a-b' \mid>\mid b-b'\mid$, then we get $b-a>b'-a'$, contradiction. \\
\textbf{$\textit{case 2.2}$}: not Cond 2 but Cond 1 \\
According to the rules of \textbf{BackInd}(), $\exists e'$ with $L_{Act}^a(e')=L_{Act}^a(e_{1})$. Assuming $L_{k}^a(e')=c$ and $L_{k+1}^a(e')=c'$, then we have $c>a>b$, $a'>c'$ and $b'>c'$. If $\mid a-b'\mid >\mid a-a'\mid$ and $\mid a-b' \mid>\mid b-b'\mid$, then we have $(b'-a')(b'+a')>2a(b'-a')$ and $a+b>2b'$. If $b'>a'>c'$, it is trivial to get contradiction; If $c'<b'<a'$, then we have $2b'<b'+a'<2a$ and $2c>a+b>2b'>2c'$, then we have $b'<a$ and $c>c'$. If $\mid a-b' \mid>\mid c-c'\mid$, then we have $a-c>b'-c'$, contradiction; If $b'=a'$, contradiction. \\
\textbf{$\textit{case 2.3}$}: not Cond 1 but Cond 2\\
According to the rules of \textbf{BackInd}(), $\exists e'$ with $L_{Act}^a(e')=L_{Act}^a(e_{2})$. proof is similar to \textbf{$\textit{case 2.2}$}.\\
\textbf{$\textit{case 2.4}$}: neither Cond 1 nor Cond 2\\
According to the rules of \textbf{BackInd}(), $\exists e', e''$ with $L_{Act}^a(e')=L_{Act}^a(e_{1})$ and $L_{Act}^a(e'')=L_{Act}^a(e_{2})$. Assuming $L_{k}^a(e')=c$ and $L_{k+1}^a(e')=c'$,$L_{k}^a(e'')=d$ and $L_{k+1}^a(e'')=d'$, then we have $d<a<c$, $d<b$, $a'>c'$ and $c'<b'<d'$. If $\mid a-b'\mid >\mid a-a'\mid$ and $\mid a-b' \mid>\mid b-b'\mid$, then we have $(b'-a')(b'+a')>2a(b'-a')$ and $(a-b)(a+b)>2b'(a-b)$. If $a>b$ and $a'>b'$, then we have $c'<c$ and $a>b'$, and if $\mid a-b'\mid >\mid c-c'\mid$, then we $a-c>b'-c'$, contradiction; If $a<b$ and $a'>b'$ or $a>b$ and $a'<b'$, it is trivial to get contradiction; If $a<b$ and $a'<b'$, then we get $d'>d$ and $a<b'$, and if $\mid a-b'\mid >\mid d-d'\mid$, then we get $b'-d'>a-d$, contradiction.\\
Proof for second inequality is similar. We skip the details.\\
\end{proof}

Proof for Theorem \ref{provenesinleave}.
\begin{proof}
We need to prove two issues: \\
1. \begin{bf}NESinLeave\end{bf}($D$) is terminated.\\
The way to prove termination of \begin{bf}NESinLeave\end{bf}($D$) is to prove $\exists k$ that after kth iteration, $\forall G_{\textit{i}},~Pp_{\textit{k}}(G_{\textit{i}})= Pp_{\textit{k+1}}(G_{\textit{i}})$.
According to Lemma \ref{provenesterminated}, trivial;\\
2. The result of \begin{bf}NESinLeave\end{bf}($D$) is $\textit{NES}$ of $D$. $\forall G_{i}\in \mathcal{V}(D)$, assuming $\pi_{\textit{i}}$ whose first edge is $e$ is the execution of $G_{\textit{i}}$ based on the result obtained by \begin{bf}NESinLeave\end{bf}($D$), we need to prove $\pi_{\textit{i}}^e$ is $\textit{NEE}$ of $G_{\textit{i}}$ coinductively. As $\pi_{\textit{i}}^e$ is ended by a cycle, we just need to prove any $e'$ on $\pi_{\textit{i}}^e$, $e'\in E(G_{\textit{j}})$, is the first edge of \textit{NEE} of $G_{\textit{j}}$. We prove edge $e$ of $G_{\textit{i}}$ as example. If $\pi_{\textit{i}}^e$ is not \textit{NEE} of $G_{\textit{i}}$, according to the definition of \textit{NEE}, there exists $\pi_{\textit{i}}^{e'}$ satisfying:
(1) $PF^d(\pi_{\textit{i}}^{e'})> PF^d(\pi_{\textit{i}}^{e})$ where $L_{\textit{Act}}^a(e')= L_{\textit{Act}}^a(e)$ or
(2) $PF^a(\pi_{\textit{i}}^{e'})> PF^a(\pi_{\textit{i}}^{e})$ where $e'=\arg\underset{e''\in E^a(G_{\textit{i}},e')}{\max} PF^d(\pi_{\textit{i}}^{e''})$, and both of them are contradicted with the rules in \textbf{BackInd}().\\
\end{proof}

\section{Tables of case study}
To make paper self-contained, we list the data related in example created in \cite{klye}.
\newcounter{Rownumber}
\newcommand{\Rown}{\stepcounter{Rownumber}\theRownumber}
\begin{table}[htbp]
\scriptsize
\centering
\begin{tabular}{cc}\hline
State number & State name\\\hline
\Rown        & $Normal\_operation$\\
\Rown        & $Httpd\_attacked$\\
\Rown        & $Ftp\_attacked$\\
\Rown        & $Ftpd\_attacked_detector$\\
\Rown        & $Httpd\_hacked$\\
\Rown        & $Ftpd\_hacked$\\
\Rown        & $Website\_defaced$\\
\Rown        & $Websever\_sniffer$\\
\Rown        & $Websever\_sniffer\_detector$\\
\Rown        & $Websever\_DOS\_1$\\
\Rown        & $Websever\_DOS\_2$\\
\Rown        & $Network\_shutdown$\\
\Rown        & $Filesever\_hacked$\\
\Rown        & $Filesever\_data\_stolen$\\
\Rown        & $Workstation\_hacked$\\
\Rown        & $Workstation\_data\_stolen\_1$\\
\Rown        & $Filesever\_data\_stolen$\\
\Rown        & $Workstation\_data\_stolen\_2$\\
\hline
\end{tabular}
\caption{\label{state}Network state}
\end{table}

\newcounter{Rownumbera}
\newcommand{\Rowna}{\stepcounter{Rownumbera}\theRownumbera}
\renewcommand{\multirowsetup}{\centering}
\begin{table}[htpb]
\scriptsize
\centering
\begin{tabular}{cccc}\hline
$State~no.\backslash$  &  1      &  2       & 3         \\
Action no.             &         &          &           \\\hline
\Rowna                 & $\phi$  &  $\phi$  & $\phi$    \\
\Rowna                 & $\phi$  &  $\phi$  & $\phi$    \\
\Rowna                 & $Install\_Sniffer\_$  & $\phi$ &   $\phi$\\
                       & $Detctor$   &    &   \\
\Rowna                 & $Remove\_Sniffer\_$  &  $\phi$     &   $\phi$\\
                       & $Detctor$ &   &  \\
\Rowna                 & $Remove\_Compromised\_$  &  $Install\_sniffer\_$  &   $\phi$\\
                       & $account\_restart\_httpd$  & $detector$ &    \\
\Rowna                 & $Remove\_Compromised\_$  &  $Install\_sniffer\_$  &   $\phi$\\
                       & $account\_restart\_ftpd$  & $detector$ &  \\
\Rowna                 & $Restore\_website\_remove\_$  &  $\phi$     &   $\phi$\\
                       & $compromised\_account$  &   &   \\
\Rowna                 & $\phi$  &  $\phi$     &   $\phi$\\
\Rowna                 & $Remove\_sniffer\_and\_$  &  $\phi$     &   $\phi$\\
                       & $Compromised\_account$  &     &   \\
\Rowna                 & $Remove\_virus\_and\_$  &  $\phi$     &   $\phi$ \\
                       & $Compromised\_account$  &   &   \\
\Rowna                 & $Remove\_virus\_and\_$  &  $\phi$     &   $\phi$\\
                       & $Compromised\_accoun$  &   &   \\
\Rowna                 & $Remove\_virus\_and\_$  &  $\phi$     &   $\phi$\\
                       & $Compromised\_account$  &    &   \\
\Rowna                 & $\phi$  &  $\phi$     &   $\phi$\\
\Rowna                 & $Remove\_sniffer\_and$  &  $\phi$     &   $\phi$\\
                       & $\_Compromised\_account$ &  &  \\
\Rowna                 & $\phi$  &  $\phi$     &   $\phi$\\
\Rowna                 & $Remove\_sniffer\_and$  &  $\phi$     &   $\phi$\\
                       & $\_Compromised\_account$  & & \\
\Rowna                 & $\phi$  &  $\phi$     &   $\phi$\\
\Rowna                 & $\phi$  &  $\phi$     &   $\phi$\\
\hline
\end{tabular}
\caption{\label{defenderaction} Defender's action set}
\end{table}

\newcounter{Rownumberb}
\newcommand{\Rownb}{\stepcounter{Rownumberb}\theRownumberb}
\begin{table}[htbp]
\scriptsize
\centering
\begin{tabular}{cccc}\hline
$State~no.\backslash$  &  1      &  2       & 3         \\
Action no.             &         &          &           \\\hline
\Rownb                 & $Attack\_httpd$  &  $Attack_ftpd$     &   $\phi$\\
\Rownb                 & $Continue\_$  &  $\phi$     &   $\phi$\\
                       & $attacking$  &  &  \\
\Rownb                 & $Continue\_$  &  $\phi$     &   $\phi$\\
                       & $attacking$   &  &  \\
\Rownb                 & $Continue\_$  &  $\phi$     &   $\phi$\\
                       & $attacking$   &  &  \\
\Rownb                 & $Deface\_$  &  $Install\_$     &   $\phi$\\
                       & $website$   & $sniffer$  &  \\
\Rownb                 & $Install\_$  &  $\phi$     &   $\phi$\\
                       & $sniffer$  &  &  \\
\Rownb                 & $\phi$  &  $\phi$  & $\phi$    \\
\Rownb                & $Run\_DOS$  &  $Crack\_fileserver\_$ &   $Crack\_workstation$\\
                       & $\_virus$ & $ root\_password$& $\_root\_password $\\
\Rownb                 & $\phi$  &  $\phi$  & $\phi$    \\
\Rownb                 & $\phi$  &  $\phi$  & $\phi$    \\
\Rownb                 & $\phi$  &  $\phi$  & $\phi$    \\
\Rownb                 & $\phi$  &  $\phi$  & $\phi$    \\
\Rownb                 & $Capture$  &  $\phi$  & $\phi$    \\
                       & $\_data$ &  &  \\
\Rownb                 & $Shutdown$  &  $\phi$  & $\phi$    \\
                       & $\_network$ &  &  \\
\Rownb                 & $Capture$  &  $\phi$  & $\phi$    \\
                       & $\_data$ &  &  \\
\Rownb                 & $Shutdown$  &  $\phi$  & $\phi$    \\
                       & $\_network$ &  &  \\
\Rownb                 & $\phi$  &  $\phi$  & $\phi$    \\
\Rownb                 & $\phi$  &  $\phi$  & $\phi$    \\
\hline
\end{tabular}
\caption{\label{attckeraction} Attacker's action set}
\end{table}

\begin{table}[htbp]
\centering
\scriptsize
\begin{tabular}{ll}\hline
$R^{1}(1)=\begin{bmatrix}10&10&10 \\10&10&10 \\0&0&0\end{bmatrix}$           &$R^{2}(1)=-R^{1}(1)$  \\
$R^{1}(2)=\begin{bmatrix}0&0&0 \\0&0&0 \\0&0&0\end{bmatrix}$                 &$R^{2}(2)=R^{1}(2)$   \\
$R^{1}(3)=\begin{bmatrix}0&0&0 \\0&0&0 \\0&0&0\end{bmatrix}$                 &$R^{2}(3)=\begin{bmatrix}-10&-10&-20 \\-10&-10&0 \\-10&-10&0\end{bmatrix}$\\
$R^{1}(4)=\begin{bmatrix}20&10&10 \\0&0&0 \\0&0&0\end{bmatrix}$              &$R^{2}(4)=\begin{bmatrix}-20&-10&-10 \\-10&0&0 \\-10&0&0\end{bmatrix}$\\
$R^{1}(5)=\begin{bmatrix}99&50&99 \\10&0&10 \\0&10&0\end{bmatrix}$          &$R^{2}(5)=\begin{bmatrix}-99&-99&-99 \\10&10&-10 \\-10&-10&0\end{bmatrix}$\\
$R^{1}(6)=\begin{bmatrix}0&0&10 \\10&0&0 \\10&0&0\end{bmatrix}$          &$R^{2}(6)=-R^{1}(6)$\\
$R^{1}(7)=\begin{bmatrix}0&0&0 \\0&0&0 \\0&0&0\end{bmatrix}$          &$R^{2}(7)=\begin{bmatrix}-99&0&0 \\-99&0&0 \\-99&0&0\end{bmatrix}$\\
$R^{1}(8)=\begin{bmatrix}30&30&30 \\50&50&50 \\50&50&50\end{bmatrix}$          &$R^{2}(8)=-R^{1}(8)$\\
$R^{1}(9)=\begin{bmatrix}-20&0&0 \\-20&0&0 \\-20&0&0\end{bmatrix}$          &$R^{2}(9)=R^{1}(9)$\\
$R^{1}(10)=\begin{bmatrix}30&0&0 \\30&0&0 \\30&0&0\end{bmatrix}$          &$R^{2}(10)=-R^{1}(10)$\\
$R^{1}(11)=\begin{bmatrix}30&0&0 \\30&0&0 \\30&0&0\end{bmatrix}$          &$R^{2}(11)=\begin{bmatrix}-60&0&0 \\-60&0&0 \\-60&0&0\end{bmatrix}$\\
$R^{1}(12)=\begin{bmatrix}0&0&0 \\0&0&0 \\0&0&0\end{bmatrix}$          &$R^{2}(12)=\begin{bmatrix}-90&0&0 \\-90&0&0 \\-90&0&0\end{bmatrix}$\\
$R^{1}(13)=\begin{bmatrix}999&0&0 \\999&0&0 \\999&0&0\end{bmatrix}$          &$R^{2}(13)=-R^{1}(13)$\\
$R^{1}(14)=\begin{bmatrix}30&60&60 \\0&0&0 \\0&0&0\end{bmatrix}$          &$R^{2}(14)=\begin{bmatrix}-10&-60&-60 \\-20&0&0 \\-20&0&0\end{bmatrix}$\\
$R^{1}(15)=\begin{bmatrix}999&0&0 \\999&0&0 \\999&0&0\end{bmatrix}$          &$R^{2}(15)=-R^{1}(13)$\\
$R^{1}(16)=\begin{bmatrix}30&60&60 \\0&0&0 \\0&0&0\end{bmatrix}$          &$R^{2}(16)=\begin{bmatrix}-10&-60&-60 \\-20&0&0 \\-20&0&0\end{bmatrix}$\\
$R^{1}(17)=\begin{bmatrix}0&0&0 \\0&0&0 \\0&0&0\end{bmatrix}$                 &$R^{2}(17)=R^{1}(17)$   \\
$R^{1}(18)=\begin{bmatrix}0&0&0 \\0&0&0 \\0&0&0\end{bmatrix}$                 &$R^{2}(18)=R^{1}(18)$   \\
\hline
\end{tabular}
\caption{\label{rewardcost} Immediate payoff to Attacker and Defender}
\end{table}

\begin{table}[htbp]
\scriptsize
\centering
\begin{tabular}{lll}\hline
$\begin{bf}State~1\end{bf}$         &$\begin{bf}State~2\end{bf}$        &$\begin{bf}State~3\end{bf}$\\
$p(2|1,1,\cdot)=1/3$                &$p(2|2,1,\cdot)=0.5/3$             &$p(3|3,1,2)=0.5$\\
$p(3|1,1,2)=1/3$                &$p(5|2,1,\cdot)=0.5/3$             &$p(3|3,1,3)=0.5$\\
$p(1|1,3,\cdot)=1/3$                &$p(1|2,2,\cdot)=1$             &$p(6|3,1,2)=0.5$\\
$ $                &$p(1|2,3,\cdot)=1$             &$p(6|3,1,3)=0.5$\\
$ $                &$ $             &$p(4|3,1,1)=1$\\
$\begin{bf}State~4\end{bf}$         &$\begin{bf}State~5\end{bf}$        &$\begin{bf}State~6\end{bf}$\\
$p(1|4,2,1)=1$                &$p(7|5,1,3)=0.8$             &$p(8|6,1,3)=0.8$\\
$p(1|4,3,1)=1$                &$p(8|5,2,3)=0.8$             &$p(9|6,1,2)=0.8$\\
$p(3|4,1,1)=1$                &$p(9|5,1,2)=0.8$             &$p(1|6,2,3)=1$\\
$p(4|4,1,2)=1$                &$p(1|5,3,1)=1$             &$p(1|6,3,1)=1$\\
$p(4|4,1,3)=1$                &$p(1|5,1,1)=1$             &$p(6|6,2,3)=1$\\
$ $                &$ $             &$p(6|6,3,3)=1$\\
$\begin{bf}State~7\end{bf}$         &$\begin{bf}State~8\end{bf}$        &$\begin{bf}State~9\end{bf}$\\
$p(1|7,\cdot,1)=1$                &$p(10|8,1,\cdot)=1/3$             &$p(1|9,\cdot,1)=1$\\
$p(7|7,\cdot,2)=0.9$                &$p(13|8,2,\cdot)=0.3$             &$ $\\
$p(7|7,\cdot,3)=0.9$                &$p(15|8,3,\cdot)=0.3$             &$ $\\
$\begin{bf}State~10\end{bf}$         &$\begin{bf}State~11\end{bf}$        &$\begin{bf}State~12\end{bf}$\\
$p(1|10,\cdot,1)=1$                &$p(1|11,\cdot,1)=1$             &$p(1|12,1,\cdot)=1$\\
$p(11|10,\cdot,2)=0.8$                &$p(12|11,\cdot,2)=0.8$             &$p(12|12,\cdot,2)=0.9$\\
$p(11|10,\cdot,3)=0.8$                &$p(12|11,\cdot,3)=0.8$             &$p(12|12,\cdot,3)=0.9$\\
$\begin{bf}State~13\end{bf}$         &$\begin{bf}State~14\end{bf}$        &$\begin{bf}State~15\end{bf}$\\
$p(14|13,1,\cdot)=1$                &$p(12|14,1,2)=1$             &$p(16|15,1,\cdot)=1$\\
$ $                &$p(12|14,1,3)=1$             &$ $\\
$ $                &$p(17|14,2,1)=1$             &$ $\\
$ $                &$p(17|14,3,1)=1$             &$ $\\
$ $                &$p(12|14,1,1)=0.5$             &$ $\\
$ $                &$p(17|14,1,1)=0.5$             &$ $\\
$\begin{bf}State~16\end{bf}$         &$\begin{bf}State~17\end{bf}$        &$\begin{bf}State~18\end{bf}$\\
$p(12|16,1,2)=1$                &$p(17|17,\cdot,\cdot)=0.9$             &$p(18|18,\cdot,\cdot)=0.9$\\
$p(12|16,1,3)=1$                &$ $             &$ $\\
$p(18|16,2,1)=1$                &$ $             &$ $\\
$p(18|16,3,1)=1$                &$ $             &$ $\\
$p(12|16,1,1)=0.5$                &$ $             &$ $\\
$p(18|16,1,1)=0.5$                &$ $             &$ $\\
\hline
\end{tabular}
\caption{\label{statetransition} State transition probabilities}
\end{table}
$ $\\

\section{Notation Index}
\begin{table}[htbp]
\small
\begin{tabular}{ll}
$Abs,~5 $& $\rm{NESinNonLeave()},~7 $\\
$\rm{AlgNES()}, ~5$& $PF^a(\pi_{\textit{i}}), ~4$\\
$\rm{AlgSOS()}, ~5 $& $PF^d(\pi_{\textit{i}}), ~4 $\\
$\rm{BackInd()}, ~6 $& $Pp_{n}(G_{\textit{i}}),~6 $\\
$\textbf{ComModel}, ~3 $& $Ps_{n}(G_{\textit{i}}),~6$\\
$\rm{ConTS}, ~4 $& $\rm{PrePro()},~7$\\
$D ,~5$&$\rm{PreProS()},~7$\\
$execution, ~4$& $\rm{PVCCS}\textsubscript{R},~2 $\\
$L_{n}(e),~6 $&$\rm{RefN()},~6$\\
$L_{n}^{S}(e), ~6$&$\rm{RefS()},~6$\\
$Leave,~5 $& $SOE, ~5$\\
$\rm{LocSoOp()},~6 $& $SOS, ~5 $\\
$NEE,~5 $& $\rm{SOSinLeave()}, ~6$\\
$NES, ~5$& $\rm{SOSinNonLeave()}, ~7$\\
$\rm{NESinLeave()},~6 $& $\mathcal{V}(D),~5$\\
\end{tabular}
\end{table}
\end{document}